\documentclass[a4paper,UKenglish,cleveref, autoref, thm-restate,authorcolumns]{lipics-v2019}
\nolinenumbers
\usepackage{microtype}

\usepackage{amssymb}
\usepackage{verbatim}
\usepackage{tikz}
\usetikzlibrary{backgrounds,decorations.pathmorphing, decorations.shapes,shapes.geometric}
\usepackage{booktabs}
\usetikzlibrary{arrows,automata}
\usepackage{float}
\usepackage{amsmath}
\usepackage{pdflscape}
\usepackage{apxproof}

\newcommand{\intreg}{\ensuremath{\mathit{int_{\mathrm{Reg}}}}}

\newcommand{\rep}{\ensuremath{\operatorname{rep}}}
\newcommand{\dec}{\ensuremath{\operatorname{decode}}}

\newcommand{\Enc}{\ensuremath{\mathsf{Enc}}}

\newcommand{\rename}{\ensuremath{\mathit{rename}}}

\newcommand{\threshold}{\ensuremath{\mathit{pick\_threshold}}}
\newcommand{\pickmerge}{\ensuremath{\mathit{pick\_merge}}}
\newcommand{\picksep}{\ensuremath{\mathit{pick\_separate}}}

\DeclareMathOperator{\lang}{\mathcal{L}}
\DeclareMathOperator{\eword}{\varepsilon}

\newtheorem{fact}{Fact}

\newcommand{\ta}{\ensuremath{\mathtt{a}}}
\newcommand{\tb}{\ensuremath{\mathtt{b}}}

\title{Regular Intersection Emptiness of Graph~Problems: Finding a Needle in a Haystack of Graphs with the Help of Automata}

\titlerunning{Regular Intersection Emptiness of Graph Problems}

\author{Petra Wolf}{Universit\"at Trier, Fachberich IV, Informatikwissenschaften,  54296 Trier, Germany \and \url{https://www.wolfp.net/}}{wolfp@uni-trier.de}{0000-0003-3097-3906}{DFG project  
FE 560/9-1}
\author{Henning Fernau}{Universit\"at Trier, Fachberich IV, Informatikwissenschaften,  54296 Trier, Germany}{fernau@uni-trier.de}{0000-0002-4444-3220}{}

\authorrunning{P.\ Wolf, H.\ Fernau}

\Copyright{Petra Wolf, Henning Fernau}



\category{}

\relatedversion{} 

\supplement{}

\funding{}

\acknowledgements{We thank Markus L.\ Schmid for his comments on an earlier version of the manuscript that greatly improved this work.}

\EventEditors{John Q. Open and Joan R. Access}
\EventNoEds{2}
\EventLongTitle{42nd Conference on Very Important Topics (CVIT 2016)}
\EventShortTitle{CVIT 2016}
\EventAcronym{CVIT}
\EventYear{2016}
\EventDate{December 24--27, 2016}
\EventLocation{Little Whinging, United Kingdom}
\EventLogo{}
\SeriesVolume{42}
\ArticleNo{23}
\begin{document}

\maketitle

\begin{abstract}
	The \intreg-problem of a combinatorial problem $P$
	asks, given a nondeterministic automaton $M$ as input, whether the language $\lang(M)$ accepted by $M$ contains any positive instance of the problem~$P$. We consider the \intreg-problem for a number of different graph problems and give general criteria that give 
 decision procedures for these \intreg-problems. To achieve this goal, we consider a natural graph encoding so that the language of all graph encodings is regular. Then, we draw the connection between classical pumping- and interchange-arguments 
	from the field of formal language theory
	with the graph operations  induced  on the encoded graph.
	Our techniques apply among others to the \intreg-problem of well-known graph problems like \textsc{Vertex Cover} and \textsc{Independent Set}, as well as to subgraph problems, graph-edit problems and graph-partitioning problems, including coloring problems. 
\end{abstract}

\newpage

\section{Introduction, Motivation and Related Work}
%
Traditional decision problems ask, given \emph{a single} instance, if this instance satisfies a certain property. But what if we do not only face a single instance, but some (representation of) a number of instances, and we like to know if \emph{any of them} satisfies the said property?

Compact representations of finite sets of instances have already been considered in several contexts. For graph problems, one might be interested if a graph satisfying a certain property, i.e., belonging to a certain graph family, is found among the graphs being similar to a given graph, this way  combinatorially modeling, for instance, input errors. 
Graph similarity is often measured in terms of edit operations~\cite{gao2010survey}, leading to graph modification problems~\cite{bodlaender2014graph,LiuWanGuo2014,Fomin15}, which have been quite a vivid research topic in parameterized algorithms in the last decade.

Searching for a positive instance among infinitely many instances of a problem $P$ seems to be a natural generalization of this setting. But how can we represent infinite sets of instances?
If we consider regular sets of instances, this task can be formalized as checking whether a given regular language of $P$-instances (represented by a  finite automaton) and the fixed language of positive $P$-instances have a non-empty intersection.
This was the original viewpoint of the line of research introduced in~\cite{guler2018deciding,Wolf:Thesis:2018}, where this problem is called the \emph{\intreg-problem of $P$} (or \intreg($P$) for short).\footnote{Note that this problem is only well-defined if it is clear how $P$ is represented as a language, i.\,e., we have to define how $P$-instances are encoded as words.} \par
The $\intreg$-problem has  been studied independently under the name \emph{regular realizability problem} $RR(L)$, where the \emph{filter language} $L$ plays the role of problem $P$  above, i.\,e., $RR(L) = \intreg(L)$ 
(see~\cite{DBLP:journals/iandc/AndersonLRSS09,DBLP:journals/corr/Rubtsov15,abs-1105-5894,DBLP:conf/csr/TarasovV11,Vyalyi11,Vyalyi13uniJournal,VyalyiR15}),  motivated by  computational complexity questions.
In this line of research, the filter languages are closely related to computations of specific machine models. This way, the regularity of the input language is not exploited at all; the hard part of a problem is coded into regular languages consisting of single words only. Vyalyi~\cite{Vyalyi13uniJournal}  notes that these reductions `cut off almost all properties of regular languages'. 


In~\cite{DBLP:journals/iandc/AndersonLRSS09,DBLP:journals/eik/HorvathKK87,Ito1988}, $\intreg(L)$ has been studied for $L$ with low computational complexity, but which describe structural properties of words that have high relevance for combinatorics on words and formal language theory (e.g., set of primitive words, palindromes, etc.).
In this regards, (efficient) decision procedures are obtained.

In contrast to these research questions, the line of work initiated in~\cite{guler2018deciding,Wolf:Thesis:2018} focuses on classical (hard) computational problems as filter languages and respective decision procedures heavily take advantage of the regularity of the set of input instances. Investigating the \intreg-problem for NP-complete problems shows that the decidability of their \intreg-problem is not trivial, e.\,g., $\intreg(\textsc{SAT})$ is decidable~\cite{guler2018deciding}, whereas $\intreg(\textsc{Bounded Tiling)}$ is not~\cite{Wolf:Thesis:2018,DBLP:journals/corr/abs-1906-08027}.
This is particularly interesting because the original hardness proofs of {\sc SAT} and {\sc Bounded Tiling} are both given by directly encoding Turing-machine computations into a problem instance~\cite{cook1971complexity,van1997convenience}.
Even low complexity classes like LOGSPACE and P contain problems with undecidable \intreg-problems~\cite{Wolf:Thesis:2018,DBLP:journals/corr/abs-1906-08027}.
Regarding the polynomial-time solvable problem \textsc{Prime} (determine if a given number is a prime number~\cite{agrawal2004primes}) it is still an open problem whether \intreg(\textsc{Prime}) is decidable~\cite{ShallitPres}. Apparently this is even unknown for regular languages of the form $uv^*w$ for words $u, v, w$.
On the contrary, for the NP-complete \textsc{Integer Linear Programming} problem, the decidability of \intreg(\textsc{Integer Linear Programming}) has been shown in~\cite{DBLP:conf/dcfs/000219}.


Here, we focus on graph problems, which deliver a rich source of NP-complete and polynomial-time solvable combinatorial problems. 
We consider a natural encoding of graphs as edge lists, so that the set of all graph encodings is a regular set. Based on this encoding, we develop a number of general criteria that imply decidability of many \intreg-problems. This stands out from the previous studies of \intreg-problems, where only singular problems have been classified as permitting a decidable \intreg-variation. 

\section{Preliminaries}


Let $\mathbb{N} = \{1, 2, 3, \ldots\}$ and $[n] = \{1, 2, \ldots, n\}$, $n \in \mathbb{N}$. For a set $A$, by $\mathcal{P}(A)$ we denote its power set and we will identify singleton sets by their elements. 
We often use combinatorial arguments in the spirit of the pigeon hole principle; the following observation is an example.

\begin{lemma} \label{lem-inclusionbound}
Let $n\in\mathbb{N}$.
Consider $\mathcal{C}\subseteq\mathcal{P}([n])$. If $|\mathcal{C}|>n$, then there is a set $A\in\mathcal{C}$ with $A\subseteq \bigcup_{B\in\mathcal{C}, B\neq A}B$.
\end{lemma}

\begin{proof}
We prove the contraposition. Hence, consider some set system  $\mathcal{C}\subseteq\mathcal{P}([n])$ in which for any  set $A\in\mathcal{C}$,  $A\subseteq \bigcup_{B\in\mathcal{C}, B\neq A}B$ is wrong. Then, there exists a function $f:\mathcal{C}\to [n]$ that proves this, as $f(A)=a$ with $a\in A\setminus \bigcup_{B\in\mathcal{C}, B\neq A}B$.
$f$ is injective, because if $f(A)=a$, then  $a\notin B$ for any other set $B\in\mathcal{A}$, so in particular $f(B)\neq a$. Hence,  $|\mathcal{C}|\leq n$.
\end{proof}

A finite, nonempty set is also known as an alphabet.
For an alphabet $A$, $A^+$ denotes the set of non-empty words over $A$ and $A^* = A^+ \cup \{\eword\}$, where $\eword$ denotes the empty word. 
 For a word $w$ over some alphabet $A$, $|w|$ denotes its length and, for every $i \in [|w|]$, $w[i]$ denotes the $i^{\text{th}}$ symbol of $w$. Moreover, by $w[i..j]$, we denote the \emph{factor} of $w$ from symbol $i$ to symbol $j$.
 A factor $w[1..i]$ with $1 \leq i \leq |w|$ is a \emph{prefix} and a factor $w[i..|w|]$ with $1 \leq i \leq |w|$ is a \emph{suffix} of~$w$. A \emph{factorization} of $w$ is a tuple $(u_1, u_2, \ldots, u_k) \in (A^*)^k$ such that $w = u_1 u_2 \ldots u_k$; we also simply represent factorizations as the concatenation of the factors, i.\,e., in the form $u_1 u_2 \ldots u_k$ (or also $w = u_1 u_2 \ldots u_k$ to emphasize that we consider a factorization of $w$). \par
A subset $L \subseteq \Sigma^*$ is a \emph{language}. For a language $L \subseteq \Sigma^*$ and 
$k \in \mathbb{N}$, we define $L / k = \{w \in \Sigma^* \mid \exists u \in \Sigma^*: |u| = k \wedge w u \in L\}$; intuitively speaking, $L / k$ is obtained from $L$ by removing the last $k$ symbols form every word.


A nondeterministic finite automaton (NFA) is a tuple $M = (\Sigma, Q, \delta, q_0, F)$ where $\Sigma$ is a finite alphabet, $Q$ is a finite set of states, $\delta \colon Q \times \Sigma \to \mathcal{P}(Q)$ is a transition function, $q_0\in Q$ is the initial state, and $F \subseteq Q$ is a set of final states. If $q'\in \delta(q,a)$, it is sometimes more convenient to view this as a triple $(q,a,q')$ called transition. 
The transition function generalizes to words in the usual way, i.e., $\delta(q, w_1w_2\dots w_n) = \delta(\dots\delta(\delta(q, w_1), w_2)\dots, w_n)$.
It also generalizes to sets of states in the following way: For a set $P \subseteq Q$ and $\sigma \in \Sigma$ let $\delta(P, \sigma) = \bigcup_{p\in P} \delta(p, \sigma)$. 
In this way, we may always apply  functions to sets of inputs.
The language accepted by an NFA $M$ is the set $\lang(M)=\{w \in \Sigma^* \mid \delta(q_0, w) \cap F \neq \emptyset\}$. 
Sometimes, we also consider a generalized NFA, allowing words (not single letters) to lead from state to state in the transitions.
For two states $q, q' \in Q$, we also consider the NFA $M[q, q'] = (\Sigma, Q, \delta, q, \{q'\})$, yielding at most  $|Q|^2$ many regular languages  $\lang(M[q, q'])$.
For a $w \in \lang(M)$, an \emph{accepting factorization (with respect to states $q_0, q_1, \ldots, q_m$)} is any factorization $w = u_1 u_2 \ldots u_m$ such that, for every~$i$ with $1 \leq i \leq m$, $u_i \in \lang(M[q_{i-1}, q_i])$, and $q_m \in F$. Recall that 
$q_0$ is the initial state. 

In general throughout this paper, we assume the tuple $ (\Sigma, Q, \delta, q_0, F)$ associated to $M$ without further mentioning. 
Also, we assume that all states of $M$ can be reached from some initial state and may lead into some final state, i.e., $M$ is reachable and co-reachable.

NFAs characterize the class of regular languages. Another characterization that we use without further formal introduction is that of regular expressions.


Throughout the paper, we consider undirected simple graphs $G = (V, E)$, where $V$ is a finite set of \emph{vertices} and $E \subseteq \{\{u,v\}\mid u, v\in V, u \neq v\}$ is a set of undirected \emph{edges}. In particular, note that this means that there is at most one edge between two vertices. 


\begin{definition}[Regular Intersection Emptiness Problem]\ \\
	For a fixed language $P\subseteq\Sigma^*$, formalizing some decision problem, the \emph{regular intersection emptiness problem} of $P$ ($\intreg(P)$ for short) is the following problem.
	\ \\
	\textit{Given:} NFA $M=(\Sigma, Q, \delta, I, F)$.\\
	\textit{Question:} Is $\lang(M) \cap P \neq \emptyset$?
\end{definition}

We are interested in the (mere) decidability status of this family of problems, depending on $P$. Hence, we need not distinguish between the emptiness or non-emptiness question. 
Below, we will describe how graphs (and numerical bounds) are encoded. As we only consider graph problems in this paper, this also fixes $\Sigma=\{\ta,\$,\#, \text{\textoneoldstyle},\triangleright\}$ in the previous definition.


%
%

\section{Main Construction: Linking Automata and Graphs}

\subsection*{Representative Functions of Automata}

%

Let $M=(\Sigma, Q, \delta, q_0, F)$ be an NFA. A \emph{representative function} (\emph{for $M$}) is a function $\rep_M \colon Q^2 \to \mathcal{P}(\Sigma^*)$ such that, for every $q, q' \in Q$, $\rep_M(q, q')$ is a finite subset of $\lang(M[q, q'])$. Each set $\rep_M(q, q')$ is called the set of \emph{$(q, q')$-representatives}. 
By assumption, the sets $\rep_M(Q^2)=\bigcup_{q,q'\in Q}\rep_M(q, q')$ of all representatives and $\Sigma_{\rep_M}=\{\rep_{M}(q, q') \mid q, q' \in Q\}$ are finite.
The \emph{$\rep_M$-condensed version} of~$M$ is the NFA $M_{\rep_M} = (\Sigma_{\rep_M}, Q, \delta', q_0, F)$, where, for every $q, q' \in Q$, $ q'\in \delta'(q, \rep_{M}(q, q'))$ iff $\rep_{M}(q, q') \neq \emptyset$.
By $\widehat{M}_{\rep_M}$, we denote the generalized NFA (over alphabet $\Sigma$) obtained from $M_{\rep_M}$ by interpreting every transition $q'\in\delta'(q, \rep_{M}(q, q'))$ as the set of transitions $\{(q, w, q') \mid w \in \rep_{M}(q, q')\}$.
The differences between these three automata are depicted in Figure~\ref{fig:MrepM} (appendix).
With the related finite substitution $\mathit{sub}:\Sigma_{\rep_M}\to  \mathcal{P}(\Sigma^*)$ that interprets the symbol  $\rep_{M}(q, q')\in \Sigma_{\rep_M}$ as a finite subset of $\Sigma^*$, we see that $\lang(\widehat{M}_{\rep_M})=\mathit{sub}(\lang(M_{\rep_M}))$. 
Hence, we find: 

\begin{proposition}
	\label{prop:subset}
$\lang(\widehat{M}_{\rep_M}) \subseteq \lang(M)$.
\end{proposition}
\begin{lemma}\label{interchangeLemma}
Let $\rep_M$ be a representative function of the NFA $M$. Let $w \in \lang(M)$ and let $w = u_1 u_2 \ldots u_m$ be an accepting factorization of $w$ with respect to states $q_0, q_1, \ldots, q_m$. 
Then, $\rep_M(q_{0}, q_1) \cdot \rep_M(q_{1}, q_2) \cdot \ldots \cdot \rep_M(q_{m-1}, q_m) \subseteq \lang(\widehat{M}_{\rep_M})$.
\end{lemma}

\subsection*{Encodings of Graphs}
We focus on combinatorial problems involving graphs. Instances of many of them can be seen as pairs of graphs and non-negative integers. We define an encoding of such pairs in the following such that the set of all encodings forms a regular language.


\begin{definition}
	Let $\mathbb{G}$ be the set of all undirected simple graphs (without loops) and let $\Enc = \lang(\triangleright \text{\textoneoldstyle}^*\$(\triangleright \ta^*\#\,\triangleright \ta^*\$)^*)$. The function $\dec \colon \Enc \to \mathbb{G} \times \mathbb{N}$ is defined as follows: 
	$$\dec(\triangleright \text{\textoneoldstyle}^k \$ \prod^m_{i = 1}(\triangleright \ta^{p_i}\#\, \triangleright\ta^{q_i}\$)) = (G, k)\,,$$
	where $G = (V, E)$ with 
$V=\{v_{p_i},v_{q_i}\mid i\in [m]\},\ E=\{\{v_{p_i},v_{q_i}\}\mid i\in [m],v_{p_i}\neq v_{q_i}\}\,.$
\end{definition}

Note that a word $w$ from $\Enc$ can contain the factor $\triangleright \ta^{i}\#\,\triangleright \ta^{j}\$$ and the factor $\triangleright \ta^{j}\#\,\triangleright \ta^{i}\$$ at the same time, and also several occurrences of the same factor $\triangleright \ta^{i}\#\,\triangleright \ta^{j}\$$. Nevertheless, by definition, $\dec(w)$ will necessarily be a simple graph. Likewise, a factor $\triangleright \ta^{i}\#\,\triangleright \ta^{i}\$$ is possible and might yield an isolated vertex.
For some $w \in \Enc$, we call the factors of the form $\triangleright  \text{\textoneoldstyle}^i \$$ as \emph{threshold tokens}, and the factors of the form $\triangleright  \ta^i \#$ and $\triangleright  \ta^i \$$ as \emph{left} and \emph{right vertex tokens}, respectively.
We refer to a factor as a \emph{vertex token} if it does not matter whether it is a left or right vertex token. 
Every $w \in \Enc$ has a unique factorization into one threshold token and a sequence of left and right vertex tokens.

%

Observe that the set  $\mathcal{I}$ of encodings envolving only edgeless graphs is not regular, as  $\mathcal{I}\cap \lang(\$\triangleright\ta^*\#\triangleright\ta^*\$)=\{\$\triangleright\ta^i\#\triangleright\ta^i\$\mid i\in\mathbb{N}\}$. 
If $\mathbb{T}\subset\mathbb{G}$ is the set of all graphs that contain some triangle, then $\dec^{-1}(\mathbb{T}\times\mathbb{N})$ is not regular either, but there is a regular language $T\subseteq \Enc$ such that $\dec(T)=\mathbb{T}\times\mathbb{N}$.
Namely, consider $T=\Enc\cdot\{\triangleright \ta \#\, \triangleright\ta\ta\$\, \triangleright \ta \#\, \triangleright\ta\ta\ta\$\, \triangleright \ta\ta \#\, \triangleright\ta\ta\ta\$\}$. 
As a third example, consider the set $\mathbb{B}$ of all bipartite graphs. Again,  $\dec^{-1}(\mathbb{B}\times\mathbb{N})$ is not regular, but $B=\lang(\triangleright \text{\textoneoldstyle}^*\$(\triangleright (\ta\ta)^*\#\,\triangleright \ta(\ta\ta)^*\$)^*)$ satisfies  $\dec(B)=\mathbb{B}\times\mathbb{N}$.
The last two examples generalize to $c$-cliques or $c$-colorability for any fixed~$c$.

This already explains the difference between questions on the syntactic level (encodings) and on the semantic level (encoded objects, in our case mostly graphs). In particular, the regular intersection emptiness problems that we consider in the following refers to the semantic level and can hence not be solved by making use of decidability results for regular languages. For instance, in this way we cannot check if the language $\lang(M)$ of some NFA $M$ contains a description of any graph that contains some triangle by testing $\lang(M) \cap T \neq \emptyset$. What we can guarantee, however, is that any NFA $M$ talking about graph properties satisfies $\lang(M)\subseteq\Enc$, because $\Enc$ is a regular set. This is one of the reasons to choose this particular graph encoding, as it avoids making regular intersection emptiness hard just by not being able to tell if any of the words of $\lang(M)$ encodes a graph.


%


\subsection*{Token-Preserving Representative Functions}
For an NFA $M=(\Sigma, Q, \delta, q_0, F)$ with $\lang(M) \subseteq \Enc$, we say that a representative function $\rep_M$ for $M$ is \emph{token-preserving} if, for every $p, q \in Q$, $\rep_M(p, q)$ is a collection of tokens. 

\begin{fact}
If $\rep_M^1, \dots, \rep_M^s$ are all representative functions for $M$, then so is their union, given as $\rep_M(p,q)=\bigcup_{i=1}^s \rep_M^i(p,q)$. If all $\rep_M^i$ are token-preserving, then is their union.
\end{fact}

As all states are reachable as well as co-reachable and as $\lang(M) \subseteq \Enc$, we can further observe the following 
 for a token-preserving representative function $\rep_M$:

\begin{fact}\label{fact:2} (a) $\bigcup_{q\in Q}\rep_M(p, q)$ contains either only right vertex tokens or left vertex tokens or threshold tokens. (b) We expect threshold tokens only in sets $\rep_M(q_0, q)$, but then there is no threshold token in any $\rep_M(q,r)$. 
(c) If the non-empty set $\rep_M(p,q)$ contains only left vertex tokens, then any non-empty $\rep_M(p,q')$ contains only left vertex tokens and non-empty $\rep_M(q',r)$ contains only right vertex tokens, so that we can partition $Q$ into four classes $Q_{\text{threshold}}$, $Q_{\text{left vertex}}$ and $Q_{\text{right vertex}}$, depending on the type of tokens that can be read from that state, and $Q_{\text{empty}}$ if no token can be read from that state. As a boundary case, we assign all the final states to $Q_{\text{left vertex}}$, even if no token can be read from them.
\end{fact}

Because of item (c), we define $\rep_M^E(p,r)=\bigcup_{q\in Q}\rep_M(p,q)\cdot \rep_M(q,r)$ to collect all edge factors that are found in sequences of left and right vertex tokens moving from state $q$ to state~$r$. Accordingly, $\rep_M^E(Q^2)$ denotes all such edge factors. Since all considered automata~$M$ are reachable, co-reachable, and $\lang(M) \subseteq \Enc$, we have $|\rep_M^E(Q^2)| \geq |\rep_M(Q^2)|$.

\begin{lemma}
	\label{lem:finite_core}
	Let $M$ be an NFA with $\lang(M) \subseteq \Enc$ and let $\rep_M$ be a token-preserving representative function for $M$. Let $m= \left|\rep_M^E(Q^2)\right|$, 
let $n = \max\{|w| \mid w \in \rep_M(Q^2)\}$ and let $\ell =|Q|^2(m+2)m  2n + n 
$. Then,
$\dec(\{w \in \lang(\widehat{M}_{\rep_M}) \mid |w| \leq \ell\}) = \dec(\lang(\widehat{M}_{\rep_M}))$. 
\end{lemma}


\begin{proof}
There are only finitely many different tokens which can appear in a word of $\lang(\widehat{M}_{\rep_M})$ since $\rep_M$ is a token-preserving representative function. 
Each word $w \in \Enc$ can contain only one threshold token (*).  
For the vertex tokens, we have to consider the context in which the token appears in some word $w \in \lang(\widehat{M}_{\rep_M})$, i.e., we have to focus on  the edges. 

Consider some $w \in \lang(\widehat{M}_{\rep_M})$ with
$|w| > \ell$. 
It has 
some accepting factorization $w=\prod_{i=1}^r w_i$ (corresponding to a state sequence $q_0,q_1,\dots,q_r$ with $w_i\in\rep_M(q_{i-1}, q_i)$).  As $w \in \lang(\widehat{M}_{\rep_M})$ and by (*), $r$ is odd, and among the $r$ tokens, there are (a) one threshold token, (b) $\frac{r-1}{2}$ left vertex tokens and (c)  $\frac{r-1}{2}$ right vertex tokens.
The tokens under (b) and (c) form  $\frac{r-1}{2}$ many edge factors.
Moreover, $\ell \leq r\cdot   \max\{|w| \mid w \in \rep_M(Q^2)\}$ (+). Also, $q_i\in Q_{\text{left vertex}}$ iff $i$ is odd. 
Since $|w| > \ell$ and by (+), 
there must exist a pair of states $(p,q)$ and a left 
vertex token $u\in\rep_M(p,q)$, such that there are $m+2$ many even indices $1\leq  i_1<i_2<\cdots i_{m+1}\leq r$ with  $p=q_{i_j-1}\in Q_{\text{left vertex}}$, $q=q_{i_j}\in Q_{\text{right vertex}}$ and $u=w_{i_j}$ for $j\in [m+2]$.

Define $\operatorname{Edgefactors}(j)=\{ w_iw_{i+1}\mid i_j\leq i<i_{j+1},\ i\ \text{is even}\}$ 
 for $j\in [m+1]$. As $\bigcup_{ j\in [m+1]}\operatorname{Edgefactors}(j)\subseteq \rep_M^E(Q^2)$, by Lemma~\ref{lem-inclusionbound}, the set system
$\{\operatorname{Edgefactors}(j)\mid j\in [m+1]\}$ must contain a specific set $\operatorname{Edgefactors}(j)$ whose edge factors also appear in $\bigcup_{ l\in [m+1], l\neq j}\operatorname{Edgefactors}(l)$. This means that we can cut out the factor $w_{i_j}w_{i_j+1}\cdots w_{i_{j+1}-1}$ from $w$, leading to some word $w'\in \Enc$ with
  $w' \in \lang(\widehat{M}_{\rep_M})$ such that $\dec(w)=\dec(w')$, as the set of edges and hence the set of vertices is not changed.
\end{proof}

For an NFA with $\lang(M) \subseteq \Enc$ and a token-preserving representative function $\rep_M$  for~$M$, we call the set $\dec(\lang(\widehat{M}_{\rep_M}))$ the \emph{finite core} of $M$ (with respect to $\rep_M$). 


\begin{definition}
	Let $M$ be an NFA with $\lang(M) \subseteq \Enc$. 
	For every $p \in Q$, we define $K_M[q_0, p] = \lang(M[q_0, p]) \cap \lang(\triangleright \text{\textoneoldstyle}^* \$)$; 
	for every $(p, q) \in Q \times Q$, we define $V_M[p, q] = \lang(M[p, q]) \cap \lang(\triangleright  \ta^* (\# | \$))$. Further, let 
	$V^G_M[p, q] = V_M[p, q] / 1$.
	For a word $w \in \lang(M)$, a factorization $w = u_1 u_2 \ldots u_m$ is called a \emph{characteristic factorization} if $u_1 \in K_M[q_0, p]$ for some state $p$ and each $u_i$ with $2 \leq i \leq m$ is contained in $V_M[p_i, q_i]$ for some states $p_i, q_i$.
	%
\end{definition}
The following observations explain the meaning of the token sets from the definition above. Here, the assumption $\lang(M) \subseteq \Enc$ and the (co-)reachability of all states are crucial.

\begin{fact}
(a) $\bigcup_{p\in Q} K_M[q_0,p]=\{\triangleright  \text{\textoneoldstyle}^k\$\mid \exists (G,k)\in\mathbb{G}\times\mathbb{N}\, \exists w\in \lang(M):\dec(w)=(G,k)\}$. 
(b)  $V_M[q_0,q]=\emptyset$ for all  $q\in Q$. 
(c) $\bigcup_{p,q\in Q}V_M^G[p,q]=\{\triangleright\ta^i\mid \exists (G,k)\in\mathbb{G}\times\mathbb{N}\, \exists w\in \lang(M):\dec(w)=(G,k),\,G=(V,E),\,v_i\in V \}$. (d) Characteristic factorizations are accepting. 
\end{fact}

Let $M$ be an NFA with $\lang(M) \subseteq \Enc$.
For every $(p,q) \in Q^2$, let $T_M[p,q]$ be a regular set of tokens such that $T_M[p,q] \subseteq \lang(M[p,q])$.
We assume a length-lexicographic (shortlex) order on the words in $T_M[p,q]$ when referring to the smallest element of the set.
\begin{itemize}
	\item $\threshold_T \colon \mathbb{N} \times Q^2 \to \mathcal{P}(\Sigma^*)$\\
	For every $k\in\mathbb{N}$, $(p, q)\in Q^2$, define $\threshold_T(k, p,q)$ as follows: If $|T_M(p,q)|<\infty$, 
	set $\threshold_T(k, p,q) = T_M(p, q)$; else,
	pick the smallest element $w$ in $T_M[p,q]$ with $|w| \geq k$ and set  $\threshold_T(k, p,q) = \{w\}$.
	\item $\pickmerge_T \colon Q^2 \to \mathcal{P}(\Sigma^*)$\\
	For 
	$(p,q) \in Q^2$, set 
	$T_M^G[p,q] = T_M[p, q] / 1$. For $A\subseteq Q^2$, let $T_M^G[A]= \bigcap_{(p, q) \in A} T_M^G[p, q]$. 
	We first define an auxiliary function $\pickmerge_T^G \colon Q^2 \to \mathcal{P}(\Sigma^*)$ from which we then derive the function $\pickmerge_T$.
		To this end, we initially let $\pickmerge_T^G(p, q)$ be the empty set for every $(p, q) \in Q^2$. For every $A \subseteq Q^2$ and for every $(p,q)\in A$, add a smallest element from $T_M^G[A]$ to $\pickmerge^G_T(p,q)$ if $T_M^G[A]\neq\emptyset$.
	Now, we can use $\pickmerge_T^G$ to define the function $\pickmerge_T$. For every $(p, q) \in Q^2$, we define $\pickmerge_T(p, q) = \{x \in T_M[p, q] \mid \{x\} /1 \subseteq \pickmerge_T^G(p, q) \} (=T_M[p, q]\cap (\pickmerge_T^G(p, q)\cdot\Sigma))$.
	
	\item $\picksep_T \colon \mathbb{N} \times \mathbb{N} \times Q^2 \to \mathcal{P}(\Sigma^*)$\\
	We describe how we define for every fixed $s, t\in \mathbb{N}$ the function $\picksep_T(s, t, p, q)$ for each $(p, q) \in Q^2$. We begin with $\picksep_T(s, t, p, q)$ being  empty  for every $(p, q)\in Q^2$. Then, we order the sets $T_M[p, q]$ arbitrarily and define $\picksep_T(s, t, p, q)$ in this order. If $T_M[p, q]$ is finite, we set $\picksep_T(s, t, p, q) = T_M[p, q]$. If $T_M[p, q]$ is infinite, we add the first (according to a length-lexicographic ordering of $T_M[p, q]$) $s$ distinct elements $w_1, \dots, w_i, \dots, w_s$ to $\picksep_T(s, t, p, q)$ for which the encoded elements are not described by any element of a previously defined set $\picksep_T(s, t, p', q')$ and for which $|w_i|\geq t$ for all $1 \leq i \leq s$.
\end{itemize}


\begin{fact}\label{fact:thresh-merge}
From the given definitions, the following two assertions are rather straight-forward. (a) The size of $\pickmerge_T(p,q)$ is bounded by $2^{|Q|^2}$, as  we pick at most one word  for each $A \subseteq Q^2$.
(b) For every $k$ encoded by a word in $\threshold_K(0, q_0, q)$, we have $k < |Q|$.
\end{fact}
\begin{theorem}\label{thm:pick-is-rep}
	Let $M$ be an NFA with $\lang(M) \subseteq \Enc$. For every $(p,q) \in Q^2$, let $T_M[p,q]$ be a regular set of tokens such that $T_M[p,q] \subseteq \lang(M[p,q])$. Then, for fixed numerical parameters, each of the functions $\threshold_T$, $\pickmerge_T$ and $\picksep_T$ is a token-preserving representative function for $M$. 
\end{theorem}
\begin{proof}
	First, observe that for each of the mentioned functions (summarized as $\mathit{pick_X}$) $\mathit{pick_X}(p,q) \subseteq T_M[p, q]$. Clearly, $\threshold_K(p,q)$ is a finite set for all states $p,q\in Q$.
	For $\pickmerge_T(p,q)$, at most one element is picked for every set $A \subseteq Q^2$, hence the size of $\pickmerge_T(p,q)$ is bounded by $2^{|Q|^2}$ and hence finite. For fixed numerical parameters, $\picksep_T$ is either equal to the finite set $T_M[p,q]$ or it contains exactly $s$ elements. 
	As $\mathit{pick_X}(p,q)$ contains only tokens it is a token-preserving representative function.
\end{proof}
\begin{proposition}
	\label{prop:computable}
	Let $M$ be an NFA with $\lang(M) \subseteq \Enc$. For every $(p, q) \in Q^2$, $k, s, t \in \mathbb{N}$, and regular set of tokens $T_M[p, q] \subseteq \lang(M[p, q])$ the sets $\threshold_T(k, p, q)$, $\pickmerge_T(p, q)$, and $\picksep_T(s, t, p, q)$ can be computed in finite time.
\end{proposition}


\subsection*{Graph Operations}

We are now going to define a number of operations on an undirected simple graph $G=(V,E)$ in a way suitable to be modeled by pumping and interchange operations on NFAs accepting encodings of graphs.
\begin{itemize}
\item A \emph{merge operation} (with respect to $u, v \in V$) consists of the following steps: remove vertices $u$ and $v$ and all their adjacent edges; add a new vertex $[u, v]$; for every former edge $\{u, w\} \in E$ or $\{v, w\} \in E$, add the edge $\{[u, v], w\}$.

\item 
	 A \emph{rename operation} (with respect to $u,v\in V$) consists in the following steps: remove the vertex $u$ and all its adjacent edges; add~$v$ as a new vertex; for every former edge $\{u,w\} \in E$, add the edge $\{v, w\}$. 

\item A \emph{vertex-deletion operation} (with respect to $v \in V$) consists in removing the vertex $v$ from $V$ and removing all edges containing $v$ from $E$.

\item An \emph{add-leaf operation} (with respect to $v \in V$) consists in the following steps: add a new vertex $v'$ to $V$; add the edge $\{v, v'\}$ to $E$.

\item A \emph{separate operation} (with respect to $v \in V$ and $w \in V$, a vertex in the neighborhood of~$v$) consists in the following two steps: remove the edge $\{v, w\}$; add a new vertex $w'$ and add the edge $\{v, w'\}$.
\end{itemize}
Three comments should help understand these operations. (a) 
An edge-contraction is the special case of a merge operation when the two merged vertices are adjacent.
(b) A separate operation with respect to $v$ and $w$  consists in performing an edge-deletion operation on $\{v, w\}$ followed by an add-leaf operation on $v$.
(c) Obviously, all considered graph properties are preserved under rename operations, which will not be mentioned any longer in the following.

\subsection*{Connecting Representative Functions to Graph Operations}

\begin{lemma}
	\label{lem:connectionMerge}
	Let $w \in \lang(M)\subseteq\Enc$ with characteristic factorization $w = u_1 u_2 \ldots u_m$ with respect to states $q_0,q_1,  \ldots, q_m$, and let $\dec(w) = (G, k)$. 
	Let $\rep_M$ be a token-preserving representative function such that for the token sets $V_M[p,q]$ (with $p, q \in Q$) $\pickmerge_V(p, q) \subseteq \rep_M(p, q)$.
	Then, there is some $w' = u_1'u_2'\ldots u_m' \in \rep_M(q_0, q_1) \cdot \rep_M(q_1, q_2) \cdot \ldots \cdot \rep_M(q_{m-1}, q_m)$ such that $\dec(w') = (G' 
, k')$ and $G'$ can be obtained from $G$ by merge and rename operations.
\end{lemma}
\begin{proof}
	By the definition of the function $\pickmerge$, it holds that for every $A \subseteq (Q \setminus \{q_0\}) \times Q$, 
	$$\left(\bigcap_{(p, q) \in A} V_M^G[p, q] \neq \emptyset\right) \Rightarrow \left(\bigcap_{(p, q) \in A} \rep_M(p, q) \neq \emptyset\right)\,.$$
	For every set of indices $P_h = \{i \in [m] \mid \{u_i\}/1 = \{u_h\}/1\}$, let $A_h = \{(q_{i-1}, q_{i}) \mid i \in P_h \}$. Clearly, $\{P_h\mid h\in [m]\}$ is a partition of $[m]$. 
	For $i = 1$ set $u_1'$ to some element in $\rep_M(q_0, q_1)$.
	For $i \in [m]$, $i > 1$, set $u_i'$ to the element added to $\pickmerge_V(p, q)$ for the set of pairs of states $A_i$ (such that $u_i$ and $u_i'$ end with the same letter by Fact~\ref{fact:2}).
	This yields a consistent renaming of the vertex encoded as $u_i$. 
	Basically, we look at all positions in the word where the same vertex appears, identify the collection of $V_M^G$ token sets related to those positions, and replace all appearances of that vertex with the single representative chosen for that collection of $V_M^G$ token sets.
	%
	The process of renaming will not disconnect any vertices that have been previously adjacent, but it might lead to merging  distinct vertices. If for example, different vertices from the set $V_M[q,q']$ appear in one single edge each, then the renaming process  replaces all of them with the same vertex, which is the element picked for $A = \{(q, q')\}$.
\end{proof}
If we picked enough elements in the separating representative function, then swapping tokens with representatives corresponds to separate and add-leaf operations on the encoded graph. 
Let $w \in \lang(M)\subseteq\Enc$ with characteristic factorization $w = u_1 u_2 \ldots u_m$ with respect to states $q_0, q_1, \ldots, q_m$.
For $p,p'\in Q$, let 
$\operatorname{indices}(p,p')=\{i\in [m]\mid q_{i-1}=p,q_i=p'\}$ and 
$\sigma_w=\max\{|\operatorname{indices}(p,p')|\colon p,p'\in Q, |V_M[p,p']|=\infty\}$. We will use this 
notation further on. 
\begin{lemma}
	\label{lem:connectionSeparate}
	Let $\dec(w) = (G, k)$ and let $\rep_M$ be a token-preserving representative function such that for the token sets $V_M[p,q]$ (with $p, q \in Q$) for some $t \geq |Q| \in \mathbb{N}$
and $s\geq\sigma_w$, 
$\picksep_V(s, t, p, q) \subseteq \rep_M(p, q)$.
	Then, there is some $w' = u_1'u_2'\ldots u_m' \in \rep_M(q_0, q_1) \cdot \rep_M(q_1, q_2) \cdot \ldots \cdot \rep_M(q_{m-1}, q_m)$ such that $\dec(w') = (G' = (V', E'), k')$ and $G'$ can be obtained from $G$ by separate, add-leaf and rename operations.
\end{lemma}
\begin{proof}
	Since $s \geq \sigma_w$, for each $u_i$, if  $|V_M[q_{i-1},q_i]|=\infty$, the set $\rep_M(q_{i-1}, q_{i})$ contains at least as many distinct elements as the number of distinct indices $h \in [m]$ with $q_{i-1} = q_{h-1}$ and $q_{i} = q_{h}$.
	For $i \in [m]$, we find the factors $u_i'$ for increasing $i$ as follows:
	For $i = 1$ set $u_1'$ to some element in $\rep_M(q_0, q_1)$.
	If $V_M[q_{i-1}, q_{i}]$ is finite, we set $u_i' = u_i$. Otherwise, choose for $u_i'$ an element in $\picksep_V(s, t, q_{i-1}, q_{i})$ which has not been assigned for any $u_g'$ with $g<i$ before. 
	Since for infinite sets $V_M[q_{i-1}, q_{i}]$, the sets $\picksep_V(s, t, q_{i-1}, q_{i})$ are disjoint, because of $t \geq |Q|$ do not contain encoded vertices from finite $V_M$-sets, 
	and contain at least $\sigma_w$ elements,
	each $u_i'$ with $u_i' \neq u_i$ encodes a vertex which is only referenced by $u_i'$ in $w'$. We are now discussing the effect of replacing a single token $u_i$ by $u_i'\neq u_i$ on the encoded graph.
For these tokens we have four cases: the assignment of $u_i'$ corresponds to (1)~the renaming of the vertex $v_i$ encoded in $u_i$, if $v_i$ is encoded only in one token; (2) a separate operation on $v_\ell$ and $v_i$ with respect to the edge $e_i = \{v_\ell, v_i\}$, partly described by $u_i$; this happens if the edge  $e_i$ is described only once in the encoding; 
	(3) an add-leaf operation on the neighbor $v_\ell$ of $v_i$ with respect to the edge $e_i = \{v_\ell, v_i\}$, partly described by $u_i$ (the edge $\{v_\ell, v_i\}$ is not removed from the graph as it might have multiple appearances in the encoding);
	(4) an add-leaf operation on $v_i$ if $(u_i/ 1)\#(u_i/ 1)\$$ forms an edge factor  (and might correspond to an isolated vertex $v_i$).
If $i_1<i_2<\ldots i_r$ are the indices with $u_{i_j}\neq u_{i_j}'$, then with $G_0=G$, by following one of the four cases described above, we arrive at a sequence of graphs $G_0,G_1,\ldots,G_r$, where $G_{j}$ is obtained from $G_{j-1}$ by executing the graph operation corresponding to the replacement of $u_{i_j}$ by $u_{i_j}'$. Observe that $G_r=G'$.
	 Since all other tokens ($u_1$ and tokens from finite $V_M$-sets) remain unchanged, the resulting graph~$G'$ can be obtained from $G$ by separate, add-leaf and rename operations.
\end{proof}
The impact of replacing tokens $u_i$ by representatives, as in the two previous lemmas, on the encoded graph is illustrated in Figure~\ref{fig:con-graph-rep} in the appendix. Note that the composition of individual replacements (graph-operations) might lead to further graph modifications as depicted in Figure~\ref{fig:graph-ops-combi} (appendix). 

\begin{lemma}
	\label{lem:conDeletion}
	Let $\lang(M) \subseteq \Enc$, $w \in \lang(M)$,  $\dec(w) = (G, k)$ with characteristic factorization $w = u_1 u_2 \ldots u_m$ with respect to the states $q_0, q_1, \ldots, q_m$.  Then, there exists a subsequence $q_0, q_{1}, q_{i_1}, q_{i_2}, \ldots, q_{i_\ell}$ of these states and a word $w'\in \lang(M)$, $\dec(w') = (G', k)$ such that $w' = u_1 u_{i_1} u_{i_2} \dots u_{i_\ell}$ is a characteristic factorization with respect to $q_0, q_1, q_{i_1}, q_{i_2}, \ldots, q_{i_\ell}$; $\ell \leq 2(|Q|-1)$ ; and $G'$ can be obtained from $G$ by edge- and vertex-deletion operations.
\end{lemma}
\begin{proof}
		If $m \leq 2|Q|-1$ the claim follows with $w = w'$, hence assume $m \geq 2|Q|$. 
		Slightly abusing the notation in Fact~\ref{fact:2}
		we collect in $Q_\text{left vertex}$ all states of $Q$ from which a left vertex token can start. Note that according to $\lang(M) \subseteq \Enc$ no factor which does not include a left vertex token as a prefix can start in a state in $Q_\text{left vertex}$ and $|Q_\text{left vertex}| \leq |Q|-2$.
		Since the factors $u_i$ are alternately left and right vertex tokens a word containing $m \geq 2|Q|$ tokens contains at least $|Q|-1$ left vertex tokens and hence there are indices $j, j' \in [m]$ such that $q_j = q_j'$ and $q_j \in Q_\text{left vertex}$. Removing the factor $u_{j+1} u_{j+2} \dots u_{j'}$ (corresponding to the state sequence $q_j, q_{q+1}, q_{q+2}, \dots, q_{j'-1}, q_{j'}$) consisting in a sequence of pairs of left and right vertex tokens from $w$ yields a word $w'' \in \lang(M)$ corresponding to the sequence of states $q_0, q_1, \dots, q_j, q_{j'+1}, \dots, q_m$. The deletion of a factor read between two states in $Q_\text{left vertex}$ corresponds to the deletion of the edges listed in the factor (and to a vertex deletion if the only tokens referring to a certain vertex were in the removed factor).
		Iteratively removing factors of this form yields the sought word $w'$ containing less than $2|Q|$ tokens for which the encoded graph $G'$ can be obtained from $G$ by edge- and vertex-delete operations. 
\end{proof}
\section{Applications -- Decidability Results}
\label{sec:appl}
After having laid the grounds for techniques essential for proving decidability of $\intreg$-problems, we now show how to apply these with two prominent graph problems: {\sc Vertex Cover} and {\sc Independent Set}.

\begin{definition}[{\sc Vertex Cover} or \textsc{VC} for short]
	\ \\
	\textit{Given:} Graph $G = (V, E)$ and a non-negative integer $k$.\\
	\textit{Question:} Is there a vertex cover (VC for short) for $G$ of size $k$ or less, i.e., a subset $V' \subseteq V$ with $|V'| \leq k$ such that for each edge $\{u,v\} \in E$ it holds that $u \in V' \vee v \in V'$?
\end{definition}
The following is a well-known property of vertex covers.
\begin{lemma}
	\label{lem:mergeVC}
	Let $G = (V, E)$ be a graph. Let $G^\star=(V^\star, E^\star)$ be obtained from $G$ by applying the merge operation on some arbitrary vertices $v$ and $v'$. If $G$ contains a VC of size at most~$k$, then $G^\star$ also contains a VC of size at most $k$.
	
\end{lemma}

\begin{lemma}\label{lem:finiteCorePosInstanceVC}
	Let $M$ be an NFA with $\lang(M) \subseteq \Enc$. Define, for $p, q \in Q$, $\rep_M(p, q)= \threshold_K(2^{|Q|^2}, p, q)\, \cup\, \pickmerge_V(p, q)$ for the token sets $K_M[p,q]$ and $V_M[p,q]$.
	Then, $\lang(M)$ contains an encoded positive \textsc{VC}-instance if and only if the finite core of $M$ (with respect to $\rep_M$) contains a positive \textsc{VC}-instance.
\end{lemma}
\begin{proof}
	The only-if-direction follows directly from Proposition~\ref{prop:subset}, as the finite core of $M$ is 
$\dec(\lang(\widehat{M}_{\rep_M}))$.
	For the if-direction, assume $w \in \lang(M)$ encodes a positive \textsc{VC}-instance $\dec(w) = (G,k)$. By the definition of $\rep_{M}$, we can use Lemma~\ref{lem:connectionMerge} to obtain a word $w' \in \lang(\widehat{M}_{\rep_M})$ with $\dec(w') = (G', k')$ such that $G'$ can be obtained from $G$ by merge and rename operations. 
	For the sets $V_M[p,q]$, the $\pickmerge$ function is used to choose representatives, hence the total number of different vertices appearing in a graph in the finite core of $M$ is bounded by $2^{|Q|^2}$, see Theorem~\ref{thm:pick-is-rep}. 
		By $\threshold_K(2^{|Q|^2}, q_0, q)$, we either have $k' > 2^{|Q|^2}$ or $k' = k$. In the former case, $G'$ trivially contains a vertex cover of size at most $k'$; in the latter case, iteratively applying Lemma~\ref{lem:mergeVC} transfers the VC for $G$ to a VC of~$G'$ which is also of size at most $k' = k$.
\end{proof}
\begin{theorem}
	$\intreg(\textsc{Vertec Cover})$ is decidable.
\end{theorem}
\begin{proof}
	Let $M$ be an NFA and an instance of $\intreg(\textsc{VC})$. As the regular languages are closed under intersection, we can assume that $\lang(M) \subseteq \Enc$.
	According to Proposition~\ref{prop:computable}, the representative function $\rep_{M}$ is computable and so is the automaton $\widehat{M}_{\rep_M}$. According to Lemma~\ref{lem:finite_core}, the finite core of $M$ is equal to $\dec(\{w \in \lang(\widehat{M}_{\rep_M}) \mid |w| \leq \ell\})$. Hence, we can enumerate all words in the finite core in finite time. Lemma~\ref{lem:finiteCorePosInstanceVC} states that $\lang(M)$ contains a positive \textsc{VC}-instance if and only if the finite core of $M$ contains a positive \textsc{VC}-instance. 
Hence, we can decide the $\intreg(\textsc{VC})$-instance $M$ 
by solving every encoded \textsc{VC}-instance in the finite core as $\textsc{VC}\in$~NP.
\end{proof}


\begin{definition}[{\sc Independent Set} or \textsc{IS} for short]
	\ \\
	\textit{Given:} A graph $G=(V, E)$ and a non-negative integer $k$.\\
	\textit{Question:} Does $G$ have an independent set (IS for short) $V'$ of size at least~$k$, i.e., is there a set $V' \subseteq V$ with $|V'| \geq k$ such that no two vertices in $V'$ are joined by an edge?
\end{definition}
\begin{lemma}
	\label{lem:ISreplace}
	Let $G = (V, E)$ be a graph, $v \in V$ with $\{v, u\} \in E$. Let $G^\star=(V^\star, E^\star)$ be obtained from $G$ by applying the separate operation on $v$ and~$u$.
	Let $G^\diamond=(V^\diamond, E^\diamond)$ be obtained from $G$ by applying the add-leaf operation on $v$.
	If $G$ contains an 
	IS of size at least~$k$, then $G^\star$ and $G^\diamond$ contain an IS of size at least~$k$.
\end{lemma}

\begin{lemma}\label{lem:finiteCoreIS}
	Let $M$ be an NFA with $\lang(M) \subseteq \Enc$. For every $p, q \in Q$, define $\rep_M(p, q)= \threshold_K(0, p, q)\, \cup\, \picksep_V(|Q|+1, |Q|, p, q)$ for the token sets $K_M[p,q]$ and $V_M[p,q]$.
	Then, $\lang(M)$ contains an encoded positive \textsc{IS}-instance if and only if the finite core of $M$ (with respect to $\rep_M$) contains a positive \textsc{IS}-instance.
\end{lemma}
\begin{proof}
	The only-if-direction follows directly from Proposition~\ref{prop:subset} since the finite core of $M$ is the set  $\dec(\lang(\widehat{M}_{\rep_M}))$.
	For the if-direction, assume $w \in \lang(M)$ encodes a positive \textsc{IS}-instance $\dec(w) = (G,k)$. 
	Let $w = u_1 u_2 \ldots u_m$ be a characteristic factorization with respect to states $q_1, q_2, \ldots, q_m$, where $\sigma_w$ is the maximal number of occurrences of two subsequent states $p, p'$ with $|V_M[p,p']|=\infty$, as formally defined before Lemma~\ref{lem:connectionSeparate}.
	
	If $\sigma_w \leq |Q|+1$, we can apply Lemma~\ref{lem:connectionSeparate} to obtain some $w' = u_1'u_2'\ldots u_m' \in \rep_M(q_0, q_1) \cdot \rep_M(q_1, q_2) \cdot \ldots \cdot \rep_M(q_{m-1}, q_m)$ such that $\dec(w') = (G' = (V', E'), k')$ and $G'$ can be obtained from $G$ by separate, add-leaf, and rename operations. Iteratively applying Lemma~\ref{lem:ISreplace} gives us that $G'$ also contains an IS of size at least~$k$. 
	Since $\rep_M$ picks from every set $K_M[p, q]$ the smallest element as the single representative, we have $k' \leq k$ and hence $G'$ also contains an IS of size at least $k'$. Hence, $w'$ encodes a positive \textsc{IS}-instance.
	
	If on the other hand $\sigma_w > |Q|+1$, then we find some  $w' = u_1'u_2'\ldots u_m' \in \rep_M(q_0, q_1) \cdot \rep_M(q_1, q_2) \cdot \ldots \cdot \rep_M(q_{m-1}, q_m)$ in the following way:
	Set $u_1'$ to the smallest element in $\threshold_K(0, q_0, q_1)$. 
	If $V_M[q_{i-1}, q_{i}]$ is finite, then $u_i \in  \picksep_V(|Q|+1, |Q|, q_{i-1}, q_i)=V_M[q_{i-1},q_i]$ and we set $u_i' = u_i$. For all other tokens proceed as follows:
	For $i > 1$ keep track on the already assigned elements in $\bigcup_{(p, q) \in (Q \setminus \{q_0\}) \times Q}\rep_M(p, q)$ and set $u_i'$ to some element in $\picksep_V(|Q|+1, |Q|, q_{i-1}, q_{i})$ which has not been picked yet. If there is no such element left, choose the largest element  in $\picksep_V(|Q|+1, |Q|, q_{i-1}, q_{i})$ for $u_i'$.
	Since $u_1'$ needs to encode the threshold~$k$ and has been set to the smallest value possible, we know that $k' \leq |Q|$ for the instance $\dec(w') = (G', k')$. 
	Let $p, p' \in Q$ be a pair of states with $|V_M[p, p']| = \infty$ and $\sigma_w$ appearances of the subsequent states $p, p'$ on the path induced by~$w$. Then, all elements in $\picksep_V(|Q|+1, |Q|, p, p')$ appear as some factors in $w'$ in a way that all but the longest word in the set appear exactly once in $w'$. As $\lang(M) \subseteq \Enc$, all elements in  $\picksep_V(|Q|+1, |Q|, p, p')$ encode either exclusively right or left sides of an edge 
	(see Fact~\ref{fact:2}) and hence are not adjacent. 
	Namely, as the set $V_M[p,p']$ is infinite and every token in it is of size at least $|Q|$, all elements in $\picksep_V(|Q|+1, |Q|, p, p')$ are distinct from the elements in $\bigcup_{(q, q') \in Q^2\backslash\{(p,p')\}}\rep_M(q,q')$. Hence, the $|Q|$ smallest elements in $\rep_M(p, p')$ encode vertices of degree one which are pairwise not adjacent and therefore already form an independent set of size $|Q| > k'$.
	Hence, $w'$ encodes a positive \textsc{IS}-instance.
\end{proof}
\begin{corollary}
	$\intreg(${\sc Independent Set}$)$ is decidable.
\end{corollary}

After having dealt with these two concrete sample problems, we are now ready to present more general criteria that describe situations when decidability of $\intreg$-problems follows from our previous reasoning.

\section{General Criteria}
Let $\mathfrak{P}_k$ be a graph property which might involve some parameter $k \in \mathbb{N}$. For a graph $G$ we denote with $\mathfrak{P}_k(G)$ that $G$ has property~$\mathfrak{P}_k$.
We say that $\mathfrak{P}_k$ is preserved under a graph operation 
if for a graph $G$ with $\mathfrak{P}_k(G)$, for any graph $G'$ which is obtained from $G$ 
by (iteratively) applying this operation 
it holds that $\mathfrak{P}_k(G')$.
We call $k$ a \emph{$\mathfrak{P}$-lower bound} if for every graph $G$ property 
$\mathfrak{P}_0(G)$ holds and from $\mathfrak{P}_k(G)$ follows $\mathfrak{P}_{k-1}(G)$ for $k\geq 1$. We call $k$ a \emph{$\mathfrak{P}$-upper bound} if for every graph $G = (V, E)$ for $k \geq |G| = |V|+|E|$, $\mathfrak{P}_k(G)$ holds and for every $k \in \mathbb{N}$ from $\mathfrak{P}_k(G)$ follows $\mathfrak{P}_{k+1}(G)$. We say that $k$ \emph{does not participate} in  $\mathfrak{P}$ if, for all $k\in\mathbb{N}$ and all graphs $G$, $\mathfrak{P}_0(G)\iff\mathfrak{P}_k(G)$. We call $k$ as \emph{$\mathfrak{P}$-nice} if $k$ is a $\mathfrak{P}$-lower or a $\mathfrak{P}$-upper bound, or if $k$ does not participate in  $\mathfrak{P}$. We say that $\mathfrak{P}$ has the \emph{leaf-property}
if there exists a monotonically nondecreasing function $f \colon \mathbb{N} \to \mathbb{N}$ such that every graph which contains at least $f(k)$ independent vertices of degree 
one (i.e., $f(k)$ many leaves) satisfies~$\mathfrak{P}_k$.
Clearly, one can view $\mathfrak{P}$ as a graph problem, where instance $(G,k)$ is positive if~$\mathfrak{P}_k(G)$.
\begin{theorem}\label{thm:general}
	Let $\mathfrak{P}_k$ be a decidable graph property.
	If one of the following holds, then $\intreg(\mathfrak{P})$ 
is decidable.
	(a) $\mathfrak{P}_k$ is preserved under merge operations and $k$ is $\mathfrak{P}$-nice.
	(b)~$\mathfrak{P}_k$ is preserved under separate and add-leaf operations; $k$ is a $\mathfrak{P}$-lower bound; and $\mathfrak{P}$ has the leaf-property.
	(c) $\mathfrak{P}_k$ is preserved under separate, add-leaf, edge-deletion, and vertex-deletion operations; and $k$ is  $\mathfrak{P}$-nice. 
\end{theorem}
\begin{proof}
	Let $M$ be an NFA.
	First, we indicate for each case how to define the representative function $\rep_M(p, q)$ for every $(p, q) \in Q^2$. Then, we proceed according to Lemmas~\ref{lem:finiteCorePosInstanceVC} and~\ref{lem:finiteCoreIS} to convert some $w \in \lang(M)$ encoding a positive $\mathfrak{P}$-
instance $\dec(w) = (G, k)$  into some word $w'$ in the finite core of $M$ encoding a positive instance. (Note that each graph in the finite core is represented by some word in $\lang(M)$.)
	Let $w=u_1u_2\dots u_m$ be a characteristic factorization of $w$ with respect to the sequence of states $q_0, q_1, \dots, q_m$.

	
	(a) 
	If $k$ is a $\mathfrak{P}$-lower bound or $k$ does not participate in $\mathfrak{P}$, we set $\rep_M(p, q) = \threshold_K(0, p, q) \cup \pickmerge_V(p, q)$.
	By Fact~\ref{fact:thresh-merge}, every $k$ encoded by a word in $\threshold_K(0, q_0, q)$ satisfies $k < |Q|$.
	If $k$ is a $\mathfrak{P}$-upper bound, then 
	set $\rep_M(p, q) = \threshold_K(2^{2|Q|^2}+ 2^{|Q|^2}, p, q) \cup \pickmerge_V(p, q)$. 
	By Fact~\ref{fact:thresh-merge}, $\rep_M(Q^2)$ encodes at most $2^{|Q|^2}$ many different vertices which form at most $2^{2|Q|^2}$ different edges.
	%
	%
	By the definition of $\rep_M$, we can use Lemma~\ref{lem:connectionMerge} to obtain a word $w' \in \lang(\widehat{M}_{\rep_M})$ with $\dec(w') = (G', k')$, $G' = (V', E')$ such that $G'$ can be obtained from $G$ by merge and rename operations. Since $\mathfrak{P}_k$ is preserved under merge operations, $\mathfrak{P}_k(G')$ follows from $\mathfrak{P}_k(G)$. 
	For $k'$ we have three cases: (1) $k' = k$, in which case $\mathfrak{P}_{k'}(G')$ is clear; (2) $k$ is a $\mathfrak{P}$-lower bound and $k' < k$; (3) $k$ is a $\mathfrak{P}$-upper bound and $k' \geq 2^{2|Q|^2} + 2^{|Q|^2} \geq |G'|$. In (2) and (3), $\mathfrak{P}_{k'}(G')$ follows from $\mathfrak{P}_k(G')$ and the definition of $\mathfrak{P}$-lower and upper bounds.
	
	(b) Define $\rep_M(p, q) = \threshold_K(0, p, q) \cup \picksep_V(f(|Q|)+1, |Q|, p, q)$. 
	%
	%
	We make a case distinction on $\sigma_w$: If $\sigma_w \leq f(|Q|)+1$, then Lemma~\ref{lem:connectionSeparate} gives us some $w' \in \lang(\widehat{M}_{\rep_M})$ with $\dec(w') = (G',k')$ and $G'$ can be obtained from $G$ by separate, add-leaf and rename operations. As $\mathfrak{P}_k$ is preserved under these operations, $\mathfrak{P}_k(G')$ follows.
	If $\sigma > f(|Q|)+1$, then replacing the factors $u_i$ in $w$ by factors $u_i'$ in the same way as in the proof of Lemma~\ref{lem:finiteCoreIS} yields a word $w' = u_1'u_2'\ldots u_m' \in \rep_M(q_0, q_1) \cdot \rep_M(q_1, q_2) \cdot \ldots \cdot \rep_M(q_{m-1}, q_m)$ such that for $\dec(w') = (G', k')$ $G'$ contains at least $f(|Q|)$ independent vertices of degree one,  implying $\mathfrak{P}_k(G')$.
	As $k' \leq k$ and since $k$ is a $\mathfrak{P}$-lower bound, we get $\mathfrak{P}_{k'}(G')$ in both cases.
	
	(c) If $k$ is a $\mathfrak{P}$-upper bound, then define $\rep_M(p, q) = \threshold_K(4|Q|^6+2|Q|^3, p, q) \cup \picksep_V(2|Q|, |Q|, p, q)$. 
	Note that if $V[p, q]$ is finite, at most $|Q|$ elements are assigned to $\picksep_V(2|Q|, |Q|, p, q)$. Hence, there are at most $2|Q|^3$ different vertices encoded in $\rep_M(Q^2)$ forming at most $4|Q|^6$ edges.
	If $k$ is a $\mathfrak{P}$-lower bound or $k$ does not participate in $\mathfrak{P}$, define $\rep_M(p, q) = \threshold_K(0, p, q) \cup \picksep_V(2|Q|, |Q|, p, q)$. 
	%
	%
	If $\sigma_w \leq 2|Q|$, we obtain the claim as in~(b), based on  Lemma~\ref{lem:connectionSeparate}. Now, assume $\sigma_w > 2|Q|$. 
	We apply Lemma~\ref{lem:conDeletion} to obtain some word $\hat{w} \in \lang(M)$ which contains at most $2|Q|$ tokens and for which $\dec(\hat{w}) = (\hat{G}, k)$ and $\hat{G}$ can be obtained from $G$ by edge- and vertex-deletion operations. 
	As $\mathfrak{P}_k$ is preserved under this operation, we get 
	$\mathfrak{P}_k(\hat{G})$ (note that we did not change the token encoding $k$).
	As the number of tokens in $\hat{w}$ is at most $2|Q|$, this implies $\sigma_{\hat{w}} \leq 2|Q|$ and 
	we can use Lemma~\ref{lem:connectionSeparate} to obtain some $w' \in \lang(\widehat{M}_{\rep_M})$ with $\dec(w') = (G',k')$ such that $G'$ can be obtained from $\hat{G}$ via separate, add-leaf and rename operations, which gives us $\mathfrak{P}_k(G')$ by our assumptions. 
	It remains to consider $k'$. 
	For $k'$, we either have $k' = k$; $k$ is a $\mathfrak{P}$-upper bound and $k' \geq 4|Q|^6+2|Q|^3 \geq |G'|$; or $k$ is a $\mathfrak{P}$-lower bound and $k' \leq k$. In all cases we get $\mathfrak{P}_{k'}(G')$.

	As in each case the definition of $\rep_M$ is constructive and as we can enumerate the finite core of $M$, we can decide $\intreg(\mathfrak{P})$ 
by testing $\mathfrak{P}_k(G)$ for each 
$(G, k)$ in the finite core. 
\end{proof}

\section{Conclusions}

We showed how to determine if within a (potentially) infinite set of instances of a graph problem, (at least) one graph with a particular property exists. 
This approach offers new connections between the area of formal languages (as we need finite descriptions for the mentioned infinite sets of instances) and graph theoretic problems.
We focused on regular languages (as the basic class of languages where many algorithmic problems are still decidable) and on a specific (but natural) encoding of graphs in the form of edge lists.
As the regular languages are closed under rational transductions (i.e., transformations defined by finite automata), similar decidability results hold for encodings obtainable by such transductions.

Let us summarize by presenting in Table~\ref{tab:corollarProblems} a list of well-known graph problems where we can conclude decidability results with our main decidability result in Theorem~\ref{thm:general}. 
In the appendix, we have collected quite a number of additional graph problems, proving the applicability of our approach to different types of problems. We mention again that it is not obvious that
$\intreg(P)$ is decidable for polynomial-time solvable problems $P$; examples in the table comprise \textsc{Connectedness}, \textsc{Emptiness} (edge-less),  \textsc{Forest} (acyclic).
\begin{table}[b]
	\centering
\scalebox{.95}{\begin{tabular}{lm{12.8cm}}
	Case& Covered Problems\\\toprule
	(a) & 
	\textsc{Connectedness},
	\textsc{Connected Vertex Cover},
	\textsc{Connected Dominating Set},	
	\textsc{Diameter},	
	\textsc{Dominating Set}, 
	\textsc{Emptiness}, 
	\textsc{Partition Into Connected Components},
	\textsc{VC}
	\\\midrule
	(b) & 
	\textsc{Acyclic Induced Subgraph}, 
	\textsc{Acyclic Subgraph}, 
	\textsc{Bipartite Induced Subgraph}, 	
	\textsc{Bipartite Subgraph},	
	\textsc{IS}, 
	\textsc{Irredundant Set}, 
	\textsc{MaxCut}, 
\textsc{Nonblocker}
	\\\midrule
	(c) & 
	\textsc{Bipartiteness},
	\textsc{Coloring}, 
	\textsc{Edge Bipartization}, 
	\textsc{Feedback Edge Set},		
	\textsc{Feedback Vertex Set}, 
	\textsc{Forest},
	\textsc{$k$-Coloring}, 	
	\textsc{Odd Cycle Transversal},	
	\textsc{Partition Into Forests}  
	\\
\end{tabular}}
	\caption{Graph-problems with decidable $\intreg$-problem according to Theorem~\ref{thm:general} listed by the applied case (a), (b), or (c), in alphabetical order.}
	\label{tab:corollarProblems}
\end{table}

There is another quite natural decoding (interpretation) of the language of encodings~$\Enc$ in terms of bipartite graphs. As these bipartite graphs have a fixed bipartition
(while otherwise a graph might have different 2-colorings), we call them \emph{red-blue graphs} in the following. Hence, the possible vertices are either red (of the form $r_i$) 
or blue (of the form $b_i$). This leads us to the following modification of the decoding function:

$$\dec_\text{red-blue}(\triangleright \text{\textoneoldstyle}^k \$ \prod^m_{i = 1}(\triangleright \ta^{p_i}\#\, \triangleright\ta^{q_i}\$)) = (G, k)\,,$$
	where $G = (V, E)$ with 
$V=R\cup B$ with $R=\{r_{p_i}\mid i\in [m]\}$,  $B=\{b_{q_i}\mid i\in [m]\}$, $E=\{\{r_{p_i},b_{q_i}\}\mid i\in [m]\}\,.$
As red-blue graphs are a natural model of hypergraphs, we can hence model $\intreg$-problems for hypergraphs, as well.
For instance, $\intreg(\textsc{Hitting Set})$ is decidable.
Alternatively, we can view~$\Enc$ as an encoding for directed graphs. Left vertex tokens would then denote the tail vertex of an arc, while right vertex tokens denote the target vertex.
Some first results on both interpretations can be found in the appendix.
Notice that in each of these new interpretations of our encoding, we need variations on the graph operations and hence on Theorem~\ref{thm:general}.
We leave it for future work to look into these interpretations in more detail, and also into  interpreting $\Enc$ as (directed) multi-graphs.
\newpage


\clearpage
\bibliography{mylit}
\clearpage
\appendix	
\section{Application of Theorem 21}

In this section, we are collecting both the definitions of quite a number of graph problems~$P$ and discuss if they possess a decidable $\intreg(P)$-problem, using the techniques presented above.
We also mention the classical complexity status of each problem~$P$.

\subsection*{Simple Basic Problems}

We start our discussion with some problems $P$ that are polynomial-time solvable on graphs.
Recall that even for such simple problems, it is not clear if $\intreg(P)$ is decidable.

\begin{definition}[{\sc Emptiness}]
	\ \\
	\textit{Given:} Graph $G = (V, E)$ (and an integer $k$).\\
	\textit{Question:} Is  $G$ empty, i.e., is $G$ edgeless?
\end{definition}

Observe that this property is maintained when merging vertices, because they are isolated. Hence, part (a) of Theorem~\ref{thm:general} applies, so that $\intreg(\textsc{Emptiness})$  is decidable.

\begin{definition}[{\sc Connectedness}]
	\ \\
	\textit{Given:} Graph $G = (V, E)$ (and an integer $k$).\\
	\textit{Question:} Is  $G$ connected, i.e., is there is a path from $u$ to $v$ within~$V'$ between any two vertices $u,v\in V'$?
\end{definition}

Again, this property is maintained when merging vertices, as merging can never create additional connected components. Hence, part (a) of Theorem~\ref{thm:general} applies, which implies that $\intreg(\textsc{Connectedness})$  is decidable.

\begin{definition}[{\sc Forest}] 
	\ \\
	\textit{Given:} Graph $G = (V, E)$ (and an integer $k$).\\
	\textit{Question:} Is  $G$ a forest, i.e., is $G$ acyclic?
\end{definition}

\begin{definition}[{\sc Bipartiteness}]
	\ \\
	\textit{Given:} Graph $G = (V, E)$ (and an integer $k$).\\
	\textit{Question:} Is  $G$ bipartite?
\end{definition}

Obviously, the integer $k$ is irrelevant in both problem definitions. Also, {\sc Forest} and {\sc Bipartiteness} can be solved in polynomial time. 
Observe that a graph stays a forest (or bipartite, resp.) after deleting edges or vertices, and also adding leaves or separating vertices will not introduce cycles (or destroy bipartiteness, resp.). 
Hence, part (c) 
of Theorem~\ref{thm:general} applies, so that $\intreg(\textsc{Forest})$ and $\intreg(\textsc{Bipartiteness})$ are decidable.

As the previous four problems have no (useful) integer parameter, we can employ one of the ideas expressed in the introduction and define a distance measure to empty graphs or to bipartite graphs: delete at most $k$ vertices in order to produce such a target graph. These problems are well-known under the names \textsc{Vertex Cover} (with respect to \textsc{Emptiness}),  \textsc{Feedback Vertex Set} (with respect to \textsc{Forest}) and \textsc{Odd Cycle Transversal} (with respect to \textsc{Bipartiteness}). We will discuss these problems below, together with an according variation of {\sc Connectedness} that we call {\sc Nearly Connected}.

The reader might have wondered why we consider {\sc Forest} rather than the (formally undefined, but seemingly simpler) problem {\sc Tree}.
However, {\sc Tree} could be seen as a combination of the basic properties underlying {\sc Forest} and {\sc Connectedness}. Therefore, neither the merging nor the separate operations preserve the tree property, i.e., our approach does not work here.

Yet, before considering these problems, let us first continue with our discussion of simple problems, now such problems which have a natural numerical parameter. We start discussing one simple problem where the status of its
$\intreg$ variant is (possibly surprisingly) unknown.

\begin{definition}[{\sc Large Vertex Degree}]
	\ \\
	\textit{Given:} Graph $G = (V, E)$ and a non-negative integer $k$.\\
	\textit{Question:} Is there a vertex  of degree $k$ or more in $G$?
\end{definition}

Clearly, \textsc{Large Vertex Degree} can be tested in polynomial time. Yet, it is not that clear at first glance if $\intreg(\textsc{Large Vertex Degree})$ is decidable.
Every graph has some vertex
of degree at least zero and if $G$ has  a vertex of degree at least $k$, it also has a vertex of degree at least $k-1$. Yet, observe that the corresponding graph property ``has a vertex of degree at least $k$'' is not preserved under merge operations, 
because edges may disappear when merging neighbors of a high-degree vertex. Hence, part (a) of Theorem~\ref{thm:general} does not apply. Part (b) and (c) does not apply either since by an arbitrary separate operation on some edge $e$ the degree of one vertex in $e$ will decrease, so that we do not know if 
 $\intreg(\textsc{Large Vertex Degree})$ is decidable and have to leave this as an open problem.

\begin{definition}[{\sc Small Vertex Degree}]
	\ \\
	\textit{Given:} Graph $G = (V, E)$ and a  non-negative integer $k$.\\
	\textit{Question:} Is there a vertex  of degree $k+1$ or less in $G$?
\end{definition}

Clearly, \textsc{Small Vertex Degree} can be tested in polynomial time. Again, the question is if $\intreg(\textsc{Small Vertex Degree})$ is decidable.
Observe that the corresponding graph property ``has a vertex of degree at most $k+1$'' is preserved under separate operations and add-leaf operations (since $k+1>0$) and moreover, every graph has some vertex
of degree at most $|V|-1$ and if $G$ has  a vertex of degree at most $k$, it also has a vertex of degree at most $k+1$. As $k+1>0$, the leaf-property holds for {\sc Small Vertex Degree}. 
Hence, part (b) of Theorem~\ref{thm:general} applies, so that indeed 
 $\intreg(\textsc{Large Vertex Degree})$ is decidable.

Even simpler decision problems belong to the graph properties ``has at least / most $k$ vertices'' or ``has at least / most $k$ edges''. Again, for each of these properties $\mathfrak{P}_k$, 
we find that $\intreg(\mathfrak{P})$ is decidable due to Theorem~\ref{thm:general}, part (a) or (b). If we want to refer to these problems explicitly in the following, we will call them \textsc{Many Vertices} / \textsc{Few Vertices} or \textsc{Many Edges} / \textsc{Few Edges}, respectively.



%
%
%

\subsection*{Large (Induced) Subgraphs}

We now consider a set of problems that can be subsumed as follows: Given a graph $G$ and a non-negative integer $k$, does there exist a set of vertices or edges of size at least~$k$ that induce a subgraph 
with a certain basic property? We have encountered one such problem before: \textsc{Independent Set} can be viewed as the problem to find a set of vertices of size at least~$k$ that induce an empty subgraph.
Clearly, for this property, the edge variant is not meaningful.

\begin{definition}[{\sc Acyclic  Subgraph}]
	\ \\
	\textit{Given:} Graph $G = (V, E)$ and a non-negative integer $k$.\\
	\textit{Question:} Does there exist a set $E' \subseteq E$ with $|E'| \geq k$ such that $G' = (V, E')$ is acyclic?
\end{definition}
\begin{definition}[{\sc Acyclic Induced  Subgraph}]
	\ \\
	\textit{Given:} Graph $G = (V, E)$ and a non-negative integer $k$.\\
	\textit{Question:} Does there exist a set $V' \subseteq V$ with $|V'| \geq k$ such that the induced graph $G[V'] = (V', E')$ is acyclic?
\end{definition}

Notice that both problems are better known in their graph edit variation (discussed below) under the names \textsc{Feedback Edge Set} and \textsc{Feedback Vertex Set}, respectively, which can be viewed as a ``dual parameterization'' of the subgraph problems we just defined. More precisely, \textsc{Feedback Edge Set} asks  if  there exists a set $E' \subseteq E$ with $|E'| \geq |E|-k$ such that $G' = (V, E')$ is acyclic, and  \textsc{Feedback Vertex Set} asks   if  there exists a set  $V' \subseteq V$ with $|V'| \geq |V|-k$ such that the induced graph $G[V'] = (V', E')$ is acyclic. This reasoning also shows that the edge variation is solvable in polynomial time, because the largest acyclic subgraph of any connected graph with $n$ vertices has $n-1$ edges and is a spanning tree; also see~\cite{DBLP:conf/coco/Karp72}. Conversely, the vertex variant is NP-complete, also see~\cite{garey1979computers,DBLP:conf/stoc/Yannakakis78}. 

For both, \textsc{Acyclic Subgraph} and \textsc{Acyclic Induced Subgraph}, we can argue that $k$ is a lower bound and that positive instances are preserved under separate and add-leaf operations.  If a graph $G$ contains at least $k$ leaves, then setting $E'$ to the set of edges incident with any leaf node yields an acyclic subgraph $G' = (V, E')$ with $|E'|\geq k$. Just taking the leaves themselves produces an induced acyclic subgraph on at least $k$ vertices. Hence, we can apply part (b) of Theorem~\ref{thm:general} in both cases.
Returning to the discussion of ``dual parameterization'' commenced above, it is interesting to note that below, we also prove decidability of the \intreg-variants of both dual problems, but that time, we will apply  part (c) of Theorem~\ref{thm:general}.

We now consider the problem to find large bipartite subgraphs; these decision problems are both known to be NP-complete; see~\cite{garey1979computers}.

\begin{definition}[{\sc Bipartite Subgraph}]
	\ \\
	\textit{Given:} Graph $G = (V, E)$ and a non-negative integer $k$.\\
	\textit{Question:} Does there exist a set $E' \subseteq E$ with $|E'| \geq k$ such that $G' = (V, E')$ is bipartite?
\end{definition}
First, note that $k$ is a {\sc Bipartite Subgraph}-lower bound and that {\sc Bipartite Subgraph} is preserved under separate and add-leaf operations, as for a vertex $v$, any  leaf $u$ added  to $v$ can be assigned to the opposite partition set (not containing $v$) within the bipartition. If a graph $G$ contains at least $k$ leaves, then setting $E'$ to the set of edges incident with any leaf node yields a bipartite graph $G' = (V, E')$ with $|E'|\geq k$ and hence part (b) applies. 

\begin{definition}[{\sc Bipartite Induced Subgraph}]
	\ \\
	\textit{Given:} Graph $G = (V, E)$ and a non-negative integer $k$.\\
	\textit{Question:} Does there exist a set  $V' \subseteq V$ with $|V'| \geq k$ such that the induced graph $G[V'] = (V, E')$ is bipartite?
\end{definition}

Basically, the same arguments as in the edge case apply, apart from the leaf-property which is now seen by considering the (empty, hence bipartite) graph induced by $k$ leaves.

Below, we will also discuss graph edit variants of both problems. Again, they can be viewed as ``dual parameterizations'', and instead of  part (b) of Theorem~\ref{thm:general}, we will apply  part (c) of Theorem~\ref{thm:general} again.

Notice that one could discuss quite a number of further problems of finding large (induced) subgraphs, but the presented problems should suffice to give the reader an idea about how the arguments work.

\subsection*{Graph Edit Problems}

Already in the introduction, we mentioned this class of problems. We are focussing here on two variations thereof: Delete at most~$k$ vertices or edges to obtain a graph with a certain property. Again,
we have seen  one such problem before: \textsc{Vertex Cover} can be viewed as the problem to find a set of vertices of size at most~$k$ whose deletion produces an empty subgraph.
The edge variant is equivalent to the polynomial-time solvable problem \textsc{Few Edges}.

We now consider the problem of deleting few edges or vertices to arrive at a bipartite graph. These problems are known to be NP-complete;  see  \cite{DBLP:journals/tcs/GareyJS76,DBLP:conf/stoc/Yannakakis78}.
 
\begin{definition}[{\sc Edge Bipartization}]
	\ \\
	\textit{Given:} Graph $G = (V, E)$ and a non-negative integer $k$.\\
	\textit{Question:} Does there exist a set $E' \subseteq E$ with $|E'| \leq k$ such that $G - E'=(V,E\setminus E')$ is bipartite?
\end{definition}
Despite the similarity of {\sc Edge Bipartization} with {\sc Bipartite Subgraph}, we cannot apply case (b) since {\sc Edge Bipartization} does not have the leaf-property.
Again, the property of containing a bipartite subgraph is maintained under separate and add-leaf operations. Since $k$ is an {\sc Edge Bipartization}-upper bound removing edges and vertices preserves the {\sc Edge Bipartization}-property and hence we can apply case (c).
\begin{definition}[{\sc Odd Cycle Transversal}]
	\ \\
	\textit{Given:} Graph $G = (V, E)$ and a non-negative integer $k$.\\
	\textit{Question:} Does there exist a set $V' \subseteq V$ with $|V'| \leq k$ such that $G - V'=G[V\setminus V]$ is bipartite?
\end{definition}
With the same considerations as for {\sc Edge Bipartization} we can apply case (c).

We now consider the same type of graph edit problems for the property ``acyclic'' instead of ``bipartite''. The complexity status of these graph edit problems was discussed above.

\begin{definition}[{\sc Feedback Vertex Set}]
	\ \\
	\textit{Given:} Graph $G = (V, E)$ and a non-negative integer $k$.\\
	\textit{Question:} Does there exist a set $V' \subseteq V$ with $|V'| \leq k$ such that $G - V'=G[V\setminus V]$ is a forest?
\end{definition}
\begin{definition}[{\sc Feedback Edge Set}]
\ \\
\textit{Given:} Graph $G = (V, E)$ and a non-negative integer $k$.\\
\textit{Question:} Does there exist a set $E' \subseteq E$ with $|E'| \leq k$ such that $G - E'=(V,E\setminus E')$ is  a forest?
\end{definition}
First, note that adding a leaf does not create a cycle. Hence, the property of being acyclic (i.e., being a forest) is preserved under separate, add-leaf, edge-deletion, and vertex-deletion operations.
 Therefore, the set of edges / vertices which have to be removed in order to make a graph acyclic will only shrink under these operations. This together with $k$ being an upper-bound for
{\sc Feedback Vertex Set} and {\sc Feedback Edge Set}, satisfies all premises for case (c). 
In the following we will see that also the directed versions of these problems have a decidable \intreg-problem.




Finally, we discuss the property ``connected''. Observe that deleting edges in order to make a graph connected is not meaningful; therefore, we only discuss the vertex variant.

\begin{definition}[{\sc Nearly Connected}]
	\ \\
	\textit{Given:} Graph $G = (V, E)$ and a non-negative integer $k$.\\
	\textit{Question:} Does there exist a set $V' \subseteq V$ with $|V'| \leq k$ such that $G - V'$ is connected?
\end{definition}

As we can determine all connected components in polynomial time, {\sc Nearly Connected} is polynomial-time solvable; also see~\cite{DBLP:conf/stoc/Yannakakis78}. 
The merge operation does not increase the number of vertices in a graph and further preserves connectedness of a graph. Hence, part (a) applies.

\subsection*{Partition Problems}

We are now considering the problem(s) of partitioning the vertex set of a graph into parts that induce graphs satisfying one of the properties ``connected'', ``acyclic'' or ``empty''.
We refrain from discussing similar edge problems here.

\begin{definition}[{\sc Partition Into Connected Components}]
	\ \\
	\textit{Given:} Graph $G = (V, E)$ and a non-negative integer $k$.\\
	\textit{Question:} Can $V$ be partitioned into $K \leq k$ disjoint sets $V_1, V_2, \dots, V_K$ such that for $1 \leq i \leq K$,  $V_i$ is connected?
\end{definition}

Recall that one can compute all connected components of a graph in polynomial time, so that we can determine in  polynomial time the smallest $k$ such that $(G,k)$ is a positive {\sc Partition Into Connected Components}-instance.
Obviously, $k$ is a {\sc Partition Into Connected Components}-upper bound, as for $k \geq |V|$ we can put each vertex in its own set. As already observed when discussing \textsc{Connected}, merging vertices can only reduce the number of connected components, so that  we can apply case (a).

\begin{definition}[{\sc Partition Into Forests}]
	\ \\
	\textit{Given:} Graph $G = (V, E)$ and a non-negative integer $k$.\\
	\textit{Question:} Can $V$ be partitioned into $K \leq k$ disjoint sets $V_1, V_2, \dots, V_K$ such that for $1 \leq i \leq K$, the subgraph induced by $V_i$ is a forest, i.e., it contains no cycles?
\end{definition}

This problem is again NP-complete; see~\cite{garey1979computers}. 
Clearly, $k$ is a {\sc Partition Into Forests}-upper bound, as for $k \geq |V|$ we can put each vertex in its own set. Further, none of the operations separate, add-leaf, edge-deletion, and vertex-deletion will produce an additional cycle and hence {\sc Partition Into Forests} is preserved under these operations and we can apply case (c).

\begin{definition}[{\sc Coloring}]
	\ \\
	\textit{Given:} Graph $G = (V, E)$ and a non-negative integer $k$.\\
	\textit{Question:} Does there exist a coloring $c \colon V \to [k]$ such that $c(u) \neq c(v)$ for every $\{u, v\} \in E$?
\end{definition}

First, it might be surprising to list this (well-known) NP-complete problem  here. Yet, one could rephrase it by asking to partition $V$ into at least $k$ subsets $V_i$ each of which induces an empty graph.
Clearly, $k$ is a {\sc Coloring}-upper bound and the property of admitting a $k$-coloring is preserved under separate, edge-deletion, and vertex-deletion operations. For $k \geq 2$, {\sc $k$-Coloring} it is also preserved under the add-leaf operation, as a leaf has only one neighbor and hence for $k \geq 2$, a color can be assigned to any additional leaf. Hence, case (c) applies.
The same type of argument works when fixing $k$ to some constant; more formally, this leads us to the following family of problems.

\begin{definition}[{\sc $k$-Coloring}]
	\ \\
	\textit{Given:} Graph $G = (V, E)$.\\
	\textit{Question:} Does there exist a coloring $c \colon V \to [k]$ such that $c(u) \neq c(v)$ for every $\{u, v\} \in E$?
\end{definition}

Observe that \textsc{2-Coloring} and \textsc{Bipartiteness} are equivalent, as well as  \textsc{1-Coloring} and \textsc{Emptiness}. 
Hence, these problems are solvable in polynomial time.
By way of contrast, {\sc $k$-Coloring} is known to be NP-complete for $k\geq 3$; see~\cite{garey1979computers}.

\subsection*{Cut Problems}


We now study cut problems, more precisely, edge cut problems. Also here, one could as well look into vertex cut problems, but this should at least clarify the flavor of these problems. 


\begin{definition}[{\sc MaxCut}]
	\ \\
	\textit{Given:} Graph $G = (V, E)$ and a non-negative integer $k$.\\
	\textit{Question:} Does there exist a partition $A \uplus B$ of $V$ such that at least $k$ edges of $G$ have one endpoint in $A$ and the second endpoint in $B$?
\end{definition}

{\sc MaxCut} is known to be NP-complete~\cite{garey1979computers}. 
Clearly, {\sc MaxCut} has the leaf-property and $k$ is a {\sc MaxCut}-lower bound. It is also clear that the {\sc MaxCut} property is preserved under separate and add-leaf operations which fits case (b).
\begin{definition}[{\sc MinCut}]
	\ \\
	\textit{Given:} Graph $G = (V, E)$ and a non-negative integer $k>0$.\\
	\textit{Question:} Does there exists a set $E' \subseteq E$ with $|E'| \leq k$ such that $G - E'$ is not connected?
\end{definition}

By the famous Max-Flow-Min-Cut theorem, this problem can be solved in polynomial time, using some flow algrorithm, also see~\cite{DBLP:conf/coco/Karp72}.

Despite $k$ being a \textsc{MinCut}-upper bound, \textsc{MinCut} has the leaf property for a constant function $f(k) = c \geq 1$ and $k \geq 1$ as we can cut the edge connecting the leaf with the rest of the graph to obtain an unconnected graph. The property of containing a minimum edge cut of size at most $k$ is further preserved under separate and add-leaf operations. Hence, we can adapt case (b) of Theorem~\ref{thm:general} for a constant function $f(k) = c \geq 1$ and an upper-bound $k$ by defining $\rep_M(p,q) = \threshold_K(((|Q|+c)|Q|^2)^2, p, q) \cup \picksep_V(c, |Q|, p, q)$. 

\subsection*{Distance-Related Graph Properties}

Recall that the distance between two vertices in an undirected graph is defined by the length of a shortest path between them. We now discuss some (only a few) graph properties that are related to this distance notion. 

\begin{definition}[{\sc $r$-Dominating Set}]
\ \\
\textit{Given:} Graph $G = (V, E)$ and a non-negative integer $k$.\\
\textit{Question:} Does there exist a set $V' \subseteq V$ with $|V'| \leq k$ such that every vertex of $G$ is within distance at most $r$ from at least one vertex of $V'$?
\end{definition}

By definition, the case $r=1$ corresponds to \textsc{Dominating Set}, which immediately entails NP-hardness.
Merging vertices will only decrease the distance of any pair of vertices in a graph since shorter paths might be created by contracting edges or merging non-adjacent vertices. As $k$ is an {\sc $r$-Dominating Set}-upper bound, we can apply case (a).

Observe that we can also consider this problem as having two numerical parameters, $r$ and~$k$. It is hence also known as $(k,r)$-\textsc{Center}. 
Observe that our reasoning also applies when fixing $k$ and considering $r$ as part of the input, a scenario often considered in approximation algorithms; see the discussions in~\cite{KATSIKARELIS201990}. The special case $k=1$ has a name of its own in graph-theoretic terminology.
\begin{definition}[{\sc Radius}]
\ \\
\textit{Given:} Graph $G = (V, E)$ and a non-negative integer $r$.\\
\textit{Question:} Is there  a vertex $c$ such that every vertex of $G$ is within distance at most~$r$ from~$c$?
\end{definition}
Notice that {\sc Radius} can be easily solved in polynomial time. Yet, our decidability result for $\intreg(\textsc{Radius})$ is not an immediate consequence of this observation, but rather follows from our reasoning. The same argument applies for the diameter instead of the radius, as merging two vertices never increases neither the radius nor the diameter of a graph.
\begin{definition}[{\sc Diameter}]
\ \\
\textit{Given:} Graph $G = (V, E)$ and a non-negative integer $d$.\\
\textit{Question:} Are all pairs of vertices of $G$ within distance at most~$d$ from each other?
\end{definition}

\subsection*{Further Graph Problems}

We first study two further main parameters of the so-called domination-chain. Both problems are NP-complete; see~\cite{garey1979computers,hedetniemi1985irredundance,DBLP:journals/siamdm/FellowsFHJ94}.
Also confer~\cite{BAZGAN2019} for a more recent survey.

\begin{definition}[{\sc Dominating Set}]
	\ \\
	\textit{Given:} Graph $G = (V, E)$ and a non-negative integer $k$.\\
	\textit{Question:} Is there a dominating set for $G$ of size $k$ or less, i.e., a subset $V' \subseteq V$ with $|V'| \leq k$ such that for all $u \in V \backslash V'$, there is a $v \in V'$ such that $\{u,v\} \in E$?
\end{definition}
Clearly, $k$ is a \textsc{Dominating Set}-upper bound, as every graph $G=(V, E)$ has a dominating set of size $\leq |V|$. Further, it is clear that the property of containing a dominating set of size $\leq k$ is maintained under merge and rename operations. Hence, case (a) of Theorem~\ref{thm:general} applies and $\intreg(\textsc{Dominating Set})$ is decidable.

\begin{definition}[{\sc Irredundant Set}]
	\ \\
	\textit{Given:} Graph $G = (V, E)$ and a non-negative integer $k$.\\
	\textit{Question:} Does there exist a set $V' \subseteq V$ with $|V'| \geq k$
	 such that $ V'$ is irredundant, i.e., each $v\in V'$ has a neighbor $u\in N[v]$ such that $N[u]\cap V'=\{v\}$?
\end{definition}
Recall that $N[v]$ denotes the closed neighborhood of $v$, i.e., the set of all vertices that are adjacent or equal to $v$.
In other words, vertices~$v$ in irredundant sets~$V'$ require a private neighbor (which could be $v$ itself), i.e., a neighbor not adjacent to any other vertex of $V'$.
Hence, for instance each inclusion-wise minimal dominating set is an irredundant set.
As also every independent set is an irredundant set, {\sc Irredundant Set} has the leaf-property with function $f(k)=k$. Moreover, $k$ is an  {\sc Irredundant Set}-lower bound. Finally, if $G$ has an irredundant set of size at least $k$, then so has any graph $G^\circ$ obtained from $G$ by a separate or add-leaf operation. More precisely, looking at Figure~\ref{fig:con-graph-rep}, both with add-leaf and with separate, a leaf $v'$ is created.
If its neighbor used to be the only private neighbor of some vertex $x$ of the irredundant set $V'$ of $G$, then $(V'\setminus\{x\})\cup\{v'\}$ is irredundant in $G^\circ$.
If the $G^\circ$-neighbor $y$ of $v'$ used to be in $V'$, then $V'$ is also irredundant in $G^\circ$, as in particular $y$ has (now) $v'$ as a private neighbor. 
If neither  $y$ nor any of the $G$-neighbors of $y$ have been in the irredundant set $V'$ of $G$, then none of the vertices of $V'$ is affected by the discussed operation, so that $V'$ is also an irredundant set in $G^\circ$.  
Our considerations cover in particular the case when a former edge $\{u,v\}$ in $G$ got replaced by an edge incident to $v'$. Hence, part (b) applies.

\begin{definition}[{\sc Monochromatic Triangle}]
	\ \\
	\textit{Given:} Graph $G = (V, E)$ (and an integer $k$).\\
	\textit{Question:} Is there a partition of $E$ into disjoint sets $E_1, E_2$ such that neither $G_1=(V, E_1)$ nor $G_2=(V, E_2)$ contains a triangle?
\end{definition}
\textsc{Monochromatic Triangle} is known to be NP-complete~\cite{garey1979computers}.
Here, $k$ does not participate in {\sc Monochromatic Triangle}. Adding leaves does not create triangle and neither does deleting edges or vertices. Hence the {\sc Monochromatic Triangle} property is preserved under separate, add-leaf, edge-deletion, and vertex-deletion operations~(case (c)).



\begin{definition}[{\sc Nonblocker}]
	\ \\
	\textit{Given:} Graph $G = (V, E)$ and a non-negative integer $k$.\\
	\textit{Question:} Is there a dominating vertex set of $G$ whose complement has at least $k$ many vertices?
\end{definition}

The complementation operation clearly does not change the classical complexity status, i.e., with \textsc{Dominating Set}, also \textsc{Nonblocker} is NP-complete. 
By adding all newly created leaves into the dominating, we see that the original nonblocker set (as the complement of a dominating set) is maintained, so that 
$k$-{\sc  Nonblocker} is preserved under separate and add-leaf operations.
Clearly,
$k$ is a {\sc  Nonblocker}-lower bound. Moreover, with $f(k)=k$, {\sc  Nonblocker} also possesses the leaf-property. Hence, $\intreg(\textsc{Nonblocker})$ is decidable by part (b).

The reader might have wondered why we do not approach the better known problem of {\sc Max-Leaf Spanning Tree}, which obviously relates to \textsc{Connected Dominating Set} that we also discuss later on. 
However,  {\sc Max-Leaf Spanning Tree} does not seem to be amenable to our approach.

\begin{definition}[{\sc $\ell$-Path Cover}]
	\ \\
	\textit{Given:} Graph $G = (V, E)$ and a non-negative integer $k$.\\
	\textit{Question:} Does there exists a set $V'\subseteq V$ with $|V'| \leq k$ such that, after removing $V'$ a graph remains where no path on $\ell$ vertices remains?
\end{definition}
For $l \geq 2$ the problem {\sc $\ell$-Path Cover} is NP-complete~\cite{DBLP:conf/stoc/Yannakakis78}. 
Notice that {\sc $2$-Path Cover} is another name for \textsc{Vertex Cover}.
Our approach only works for $\ell=2$, because by vertex merging as well as by adding leafs, longer paths can be created.

%
%

\subsection*{Connected Problem Variations}

Many graph problems can be seen as selecting a set of vertices $V'$ with  certain properties; it is possible to add further requirements, for instance, that $V'$ is (also) connected.
We discuss this also NP-complete variation (see~\cite{garey1979computers}) for some of the problems considered above.

\begin{definition}[{\sc Connected Vertex Cover}]
	\ \\
	\textit{Given:} Graph $G = (V, E)$ and a non-negative integer $k$.\\
	\textit{Question:} Is there a \emph{connected vertex cover} for $G$ of size $k$ or less, i.e., a subset $V' \subseteq V$ with $|V'| \leq k$ that is both connected and a vertex cover?
\end{definition}
In other words, $V'$ is a connected vertex cover if for each edge $\{u,v\} \in E$, we find $u \in V' \vee v \in V'$, and if between any two vertices $u,v\in V'$, there is a path from $u$ to $v$ within~$V'$.

If we look carefully at the proof of Theorem~\ref{thm:general} case (a) and Lemma~\ref{lem:finiteCorePosInstanceVC} we observe that we can relax the condition that $\mathfrak{P}_k(G)$ holds for $k \geq |G|$ to the following condition: if for a graph $G$ there is any $k \geq |G|$ such that $\mathfrak{P}_k(G)$ holds, then $\mathfrak{P}_{|G|}(G)$ holds (as we only need an upper bound on the value of $k$ above which the actual value of $k$ does not matter anymore).
As $V$ itself is always a valid connected vertex cover if $G$ is connected, this condition holds.
Further, noticing that the property of containing a connected vertex cover of size at most $k$ is preserved under merge and rename operations we obtain the decidability of $\intreg(\textsc{Connected Vertex Cover})$ analogously to case (a) of Theorem~\ref{thm:general}. 

With  an analogous argument, one can prove the decidability of the $\intreg$ variation of the following problem:

\begin{definition}[{\sc Connected Dominating Set}]
	\ \\
	\textit{Given:} Graph $G = (V, E)$ and a non-negative integer $k$.\\
	\textit{Question:} Is there a \emph{connected dominating set} for $G$ of size $k$ or less, i.e., a subset $V' \subseteq V$ with $|V'| \leq k$ that is both connected and a dominating set?
\end{definition}

One can also consider the problems \textsc{Connected Feedback Vertex / Edge Set}, but here we observe that our techniques do not apply.

\section{Beyond Simple Undirected Graphs}
%
There is another quite natural decoding (interpretation) of the language of encodings~$\Enc$ in terms of bipartite graphs. As these bipartite graphs have a fixed bipartition
(while otherwise a graph might have different 2-colorings), we call them \emph{red-blue graphs} in the following. Hence, the possible vertices are either red (of the form $r_i$) 
or blue (of the form $b_i$). This leads us to the following modification of the decoding function:

$$\dec_\text{red-blue}(\triangleright \text{\textoneoldstyle}^k \$ \prod^m_{i = 1}(\triangleright \ta^{p_i}\#\, \triangleright\ta^{q_i}\$)) = (G, k)\,,$$
	where $G = (V, E)$ with 
$V=R\cup B$ with $R=\{r_{p_i}\mid i\in [m]\}$,  $B=\{b_{q_i}\mid i\in [m]\}$, $E=\{\{r_{p_i},b_{q_i}\}\mid i\in [m]\}\,.$
There are some subtle differences between this interpretation and graphs that are just  bipartite. Most notably, in our definition, there is no encoding for red-blue graphs with isolated vertices. 
Also, with the definition of graph operations, we have to be careful.  By the previous observation, we should pay attention when deleting arbitrary vertices or edges, as this might lead to isolated vertices.
More precisely, we are now facing the following (modified) graph operations: (a) If we delete a vertex $x$ from a red-blue graph, we do not only delete $x$ and all its incident edges, but also all isolated vertices that might be created this way.
In other words, we will delete, in addition, all neighbors of $x$ that have been leaves before deleting~$x$.
(b) The same problem may occur when deleting edges: if we delete an edge~$e$, we also remove all vertices incident to~$e$ that have been of degree one.
 Also, the merge and rename operations should be color-preserving; in particular, red vertices should be merged with red vertices only. The add-leaf operation would implicitly take care of the fact that a leaf added to a blue vertex should be red and vice versa. Finally, a separate operation with respect to $v$ and $w$ adds in particular  a new vertex $w'$ of the same color as $w$.
As a technical remark, the modified edge and vertex deletion is only performed after a merge or separate operation and not as an intermediate step thereof.
To this end, more formally for every $(p, q) \in Q^2$, we define $R_M[p, q] = \lang(M[p, q])\cap \lang(\triangleright \ta^*\#)$ and $B_M[p, q] = \lang(M[p, q])\cap \lang(\triangleright \ta^*\$)$ be the set of red vertex and blue vertex tokens. Finally, when talking about the \emph{finite core} in the following, we now use the decoding function $\dec_\text{red-blue}$ instead of $\dec$.

\begin{lemma}
	\label{lem:connectionRedBlueVertexMerge}
	Let $w \in \lang(M)$ with characteristic factorization $w = u_1 u_2 \ldots u_m$ with respect to states $q_0,q_1,  \ldots, q_m$, and let $\dec_\text{red-blue}(w) = (G, k)$, 
	with $G=(V,E)$ where $V=R\cup B$. 
	Let $\rep_M$ be a token-preserving representative function such that for the token sets $R_M[p,q]$, and $B_M[p,q]$ (with $p, q \in Q$) $(\pickmerge_R(p, q) \cup \pickmerge_B(p, q))\subseteq \rep_M(p, q)$.
	Then, there is some $w' = u_1'u_2'\ldots u_m' \in \rep_M(q_0, q_1) \cdot \rep_M(q_1, q_2) \cdot \ldots \cdot \rep_M(q_{m-1}, q_m)$ such that $\dec_\text{red-blue}(w') = (G' , k')$ and $G'$ can be obtained from $G$ by color-preserving merge and rename operations.
\end{lemma}

\begin{proof}
	We proceed as in the proof of Lemma~\ref{lem:connectionMerge} except that we treat red and blue vertex  tokens separately.
	In particular, we consider the partition of $[m]$ given by $P_h=\{i\in[m]\mid u_i=u_h\}$. Depending on whether $u_i$ is a left or a right vertex token, referring to a red or a blue vertex, we are using $\pickmerge_R$ or $\pickmerge_B$ to obtain the representative $u_i'$.
\end{proof}

\begin{definition}[{\sc Red-Blue Dominating Set}, or {\sc RBDS} for short]
	\ \\
	\textit{Given:} Red-blue graph $G = (V, E)$, $V=R\cup B$, and a non-negative integer $k$.\\
	\textit{Question:} Is there a red-blue dominating set (RBDS) for $G$ of size $k$ or less, i.e., a subset $R' \subseteq R$ with $|R'| \leq k$ such that for all $b \in B$ there is a $r \in R'$ such that $\{r,b\} \in E$?
\end{definition}

As we see below in our discussions, \textsc{RBDS} is equivalent to \textsc{Hitting Set} and \textsc{Set Cover} and hence NP-complete.

\begin{lemma}
	\label{lem:mergeRBDS}
	Let $G = (V, E)$ be a red-blue graph, with $V=R\cup B$. Let $G^\star=(V^\star, E^\star)$ be obtained from $G$ by applying the color-preserving merge operation on some arbitrary vertices $v$ and $v'$. If $G$ contains an RBDS of size at most $k$, then $G^\star$ also contains an RBDS of size at most $k$.
\end{lemma}

For the proof, just observe that the size of the RBDS might drop by one if two red vertices are merged, namely, when the two merged vertices belonged to a smallest RBDS, but 
it might also drop if two blue vertices are merged, as then it might be possible to remove one of the vertices from the RBDS, as it might have lost its private neighbor, i.e., its only neighbor that was not adjacent to any other red vertex in the RBDS.
%

\begin{lemma}\label{lem:finiteCorePosInstanceRBDS}
	Let $M$ be an NFA with $\lang(M) \subseteq \Enc$. Define, for $p, q \in Q$, $\rep_M(p, q)= \threshold_K(2^{|Q|^2}, p, q)\, \cup\, \pickmerge_R(p, q)\, \cup\, \pickmerge_B(p, q)$ for the token sets $K_M[p,q]$, $R_M[p,q]$ and $B_M[p,q]$.
	Then, $\lang(M)$ contains an encoded positive \textsc{RBDS}-instance if and only if the finite core of $M$ (with respect to $\rep_M$) contains a positive \textsc{RBDS}-instance.
\end{lemma}

The proof is quite analogous to the one of Lemma~\ref{lem:finiteCorePosInstanceVC}, taking care of the peculiarities of red-blue graphs, now using Lemmas~\ref{lem:connectionRedBlueVertexMerge} and~\ref{lem:mergeRBDS}.


\begin{theorem}
	$\intreg(\textsc{RBDS})$ is decidable.
\end{theorem}

Observe that there are at least two more natural decodings (interpretations) of the language $\Enc$ of encodings that we defined above for instances of typical graph problems.
\begin{description}
	\item[Hypergraphs] $\dec_\text{hyp}(\triangleright \text{\textoneoldstyle}^k \$ \prod^m_{i = 1}(\triangleright \ta^{p_i}\#\, \triangleright\ta^{q_i}\$)) = (G, k)\,,$ where the hypergraph $G=(V,E)$ is described by the universe (vertex set) $V=\{v_{p_i}\mid i\in [m]\}$ and the hyperedge $e_{q_i}$ collects all $v_{p_j}$ such that $q_i=q_j$, yielding the hyperedge set $E$. 
	
	\item[Directed graphs] $\dec_\text{dir}(\triangleright \text{\textoneoldstyle}^k \$ \prod^m_{i = 1}(\triangleright \ta^{p_i}\#\, \triangleright\ta^{q_i}\$)) = (G, k)\,,$ where the directed graph $G=(V,E)$ is described by the vertex set $V=\{v_{p_i},v_{q_i}\mid i\in [m]\}$ and $E=\{(v_{p_i},v_{q_i})\mid i\in [m]\}$. Clearly, $\dec(w)$ delivers the underlying undirected simple graph of  $\dec_\text{dir}(w)$, obtained from the latter by forgetting arc directions and omitting loops.
\end{description}

Recall that hypergraphs can be interpreted as red-blue graphs, with the vertex set of the hypergraph collecting the red vertices and the hyperedge set collecting the blue vertices. Observe that the red-blue graph obtained from $\dec_\text{hyp}(w)$ by this interpretation equals $\dec_\text{red-blue}(w)$ and likewise, the hypergraph corresponding 
to the red-blue graph  $\dec_\text{red-blue}(w)$ equals $\dec_\text{hyp}(w)$.
Therefore, we can immediately translate our results on red-blue graph problems into results on hypergraph problems. Hence, we get \intreg-decidability for the NP-complete problem \textsc{Hitting Set}~\cite{garey1979computers}.

\begin{definition}[{\sc Hitting Set}]
	\ \\
	\textit{Given:} Collection $E$ of subsets of a finite set $V$, defining a hypergraph $(V,E)$,  and a non-negative integer $k$.\\
	\textit{Question:} Is there a subset $V' \subseteq V$ with $|V'| \leq k$ such that $V'$ contains at least one element from each hyperedge in $E$?
\end{definition}

\begin{theorem}
	$\intreg(\textsc{Hitting Set})$ is decidable.
\end{theorem}

\begin{remark}
There are other natural encodings for hypergraphs in particular. Without going into details, one possibility would be to present all hyperedges by listing their vertex tokens. This could be carried out in a way that our encodings for undirected graphs
would appear as a special case (disregarding loops).
Although the previous result is true for both encodings, observe that the encodings do not translate directly into each other, because there is no rational transducer that translates between the two hypergraph encodings.
\end{remark}

We are now turning to another well-known NP-complete problem on hypergraphs, also known as set systems~\cite{garey1979computers}.

\begin{definition}[{\sc Set Cover}]
	\ \\
	\textit{Given:} Collection $E$ of subsets of a finite set $V$, defining a hypergraph $(V,E)$,  and a non-negative integer $k$.\\
	\textit{Question:} Is there a subset $E' \subseteq E$ with $|E'| \leq k$ such that $V=\bigcup_{e\in E'}e$?
\end{definition}

If we interpret this classical problem from the viewpoint of red-blue graphs, this immediately translates into the question of finding a set of at most $k$ blue vertices that dominate all red vertices. Now, interchanging the roles of red and blue vertices, which is nothing else than applying the concept of  hypergraph duality, we immediately deduce by observing that the red and blue vertices are treated alike in all our graph operations:

\begin{theorem}
	$\intreg(\textsc{Set Cover})$	 is decidable.
\end{theorem}

Let us now turn our attention towards directed graphs, or digraphs for short. Notice that although we do not allow multiple edges (or better called arcs in this setting) in the same direction, it is usual (and also quite natural) to have loops in the interpretation of a word from $\Enc$, and also there could be an arc from $u$ to $v$ and another arc from $v$ to $u$.
In particular, if we merge two adjacent vertices $u$ and $v$, a loop on $[u,v]$ will result, and if there is an arc from $w$ to $u$ and from $v$ to $w$, then there will
be one arc in either direction between $w$ and $[u,v]$. Recall that the separate and add-leaf operations are realized by pumping parts of the encodings, which means that directions will be maintained. More specifically, if there is an arc from $u$ to $v$ and $v'$ is introduced as a copy of $v$ in an arc factor corresponding to $(u,v)$, then there will be an arc from $u$ to $v'$ in the resulting graph, while if this happens in an arc factor corresponding to $(v,u)$, then we will see an arc from $v'$ to $u$. 
In the case of an arc factor for $(v,v)$, it depends on whether the left vertex token or the right vertex token is involved in the pumping to understand if an arc from $v'$ to $v$ or vice versa is introduced.

Having these rather minor modifications in mind, basically all general lemmas and theorems that we developed in the undirected graph setting can be adapted to the directed setting. These considerations prove the decidability of $\intreg$ for the following problems:\begin{itemize}
\item
\textsc{Directed Forest}: Determine if a digraph is a collection of directed acyclic graphs. 
\item \textsc{Directed Feedback Vertex Set}, \textsc{Directed Feedback Arc Set}, \textsc{Directed Acyclic Subgraph}, \textsc{Directed Acyclic Induced Subgraph}.
\item \textsc{Diameter}, \textsc{Directed Dominating Set} 
\end{itemize}

We only give the formal definition for two NP-complete problems~\cite{DBLP:conf/coco/Karp72} of these cases below.

\begin{definition}[{\sc Directed Feedback Vertex Set}]
\ \\
\textit{Given:} Directed graph $G = (V, E)$ and a non-negative integer $k$.\\
\textit{Question:} Does there exist a set $V' \subseteq V$ with $|V'| \leq k$ such that $G - V'$ is acyclic?
\end{definition}
\begin{definition}[{\sc Directed Feedback Arc Set}]
\ \\
\textit{Given:} Directed graph $G = (V, E)$ and a non-negative integer $k$.\\
\textit{Question:} Does there exist a set $E' \subseteq E$ with $|E'| \leq k$ such that $G - E'$ is acyclic?
\end{definition}

\section{Illustrating Figures}
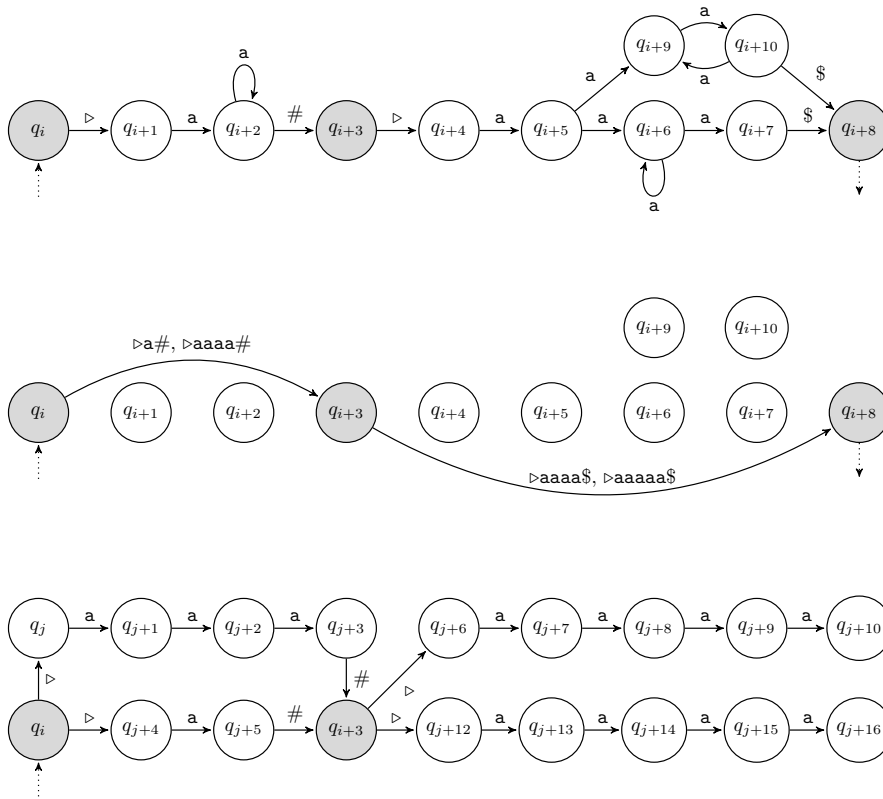
\begin{figure}[H]
	\scalebox{0.75}{
	\begin{tikzpicture}[->,>=stealth',shorten >=1pt,auto,semithick, node distance=1.8cm]
		\tikzstyle{every state} = [minimum size=30pt]
		
		\node[state, fill=gray!30] (0) at (0,0) {$q_i$};
		\node[below of=0, node distance=1.3cm] (-1) {};
		\node[state, right of=0] (1) {$q_{i+1}$};
		\node[state, right of=1] (2) {$q_{i+2}$};
		\node[state, right of=2, fill=gray!30] (3) {$q_{i+3}$};
		\node[state, right of=3] (4) {$q_{i+4}$};
		\node[state, right of=4] (5) {$q_{i+5}$};
		\node[state, right of=5] (6) {$q_{i+6}$};
		\node[state, right of=6] (7) {$q_{i+7}$};
		\node[state, right of=7, fill=gray!30] (8) {$q_{i+8}$};
		\node[state, node distance=1.5cm, above of=6] (9) {$q_{i+9}$};
		\node[state, node distance=1.5cm, above of=7] (10) {$q_{i+10}$};
		\node[below of=8, node distance=1.3cm] (11) {};
		
		\path
		(-1) edge[dotted] (0)
		(0) edge node[above] {$\triangleright$} (1)
		(1) edge node[above] {$\ta$} (2)
		(2) edge[loop above] node {$\ta$} (2)
		(2) edge node {$\#$} (3)
		(3) edge node {$\triangleright$} (4)
		(4) edge node {$\ta$} (5)
		(5) edge node {$\ta$} (6)
		(6) edge node {$\ta$} (7)
		(7) edge node {$\$$} (8)
		(9) edge[bend left] node {$\ta$} (10)
		(10) edge[bend left] node {$\ta$} (9)
		(10) edge node {$\$$} (8)
		(5) edge node {$\ta$} (9)
		(6) edge[loop below] node {$\ta$} (6)
		(8) edge[dotted] (11)
		;
	\end{tikzpicture}
}

\vspace{1cm}
\scalebox{0.75}{
	\begin{tikzpicture}[->,>=stealth',shorten >=1pt,auto,semithick, node distance=1.8cm]
	\tikzstyle{every state} = [minimum size=30pt]
	
	\node[state, fill=gray!30] (0) at (0,0) {$q_i$};
	\node[below of=0, node distance=1.3cm] (-1) {};
	\node[state, right of=0] (1) {$q_{i+1}$};
	\node[state, right of=1] (2) {$q_{i+2}$};
	\node[state, right of=2, fill=gray!30] (3) {$q_{i+3}$};
	\node[state, right of=3] (4) {$q_{i+4}$};
	\node[state, right of=4] (5) {$q_{i+5}$};
	\node[state, right of=5] (6) {$q_{i+6}$};
	\node[state, right of=6] (7) {$q_{i+7}$};
	\node[state, right of=7, fill=gray!30] (8) {$q_{i+8}$};
	\node[state, node distance=1.5cm, above of=6] (9) {$q_{i+9}$};
	\node[state, node distance=1.5cm, above of=7] (10) {$q_{i+10}$};
	\node[below of=8, node distance=1.3cm] (11) {};
	
	\path
	(-1) edge[dotted] (0)
	(0) edge[bend left] node {$\triangleright \ta \#$, $ \triangleright \ta\ta\ta\ta\#$} (3)
	(3) edge[bend right] node {$\triangleright \ta \ta \ta \ta \$$, $ \triangleright \ta \ta \ta \ta \ta \$$} (8)
	(8) edge[dotted] (11)
	;
	\end{tikzpicture}
}

\vspace{1cm}
\scalebox{0.75}{
\begin{tikzpicture}[->,>=stealth',shorten >=1pt,auto,semithick, node distance=1.8cm]
\tikzstyle{every state} = [minimum size=30pt]

\node[state,fill=gray!30] (0) at (0,0) {$q_i$};
\node[below of=0, node distance=1.3cm] (-1) {};
\node[state, right of=0] (1) {$q_{j+4}$};
\node[state, right of=1] (2) {$q_{j+5}$};
\node[state, fill=gray!30, right of=2] (3) {$q_{i+3}$};
\node[state, right of=3] (4) {$q_{j+12}$};
\node[state, right of=4] (5) {$q_{j+13}$};
\node[state, right of=5] (6) {$q_{j+14}$};
\node[state, right of=6] (7) {$q_{j+15}$};
\node[state, above of=6] (9) {$q_{j+8}$};
\node[state, right of=7] (16) {$q_{j+16}$};
\node[state, above of=7] (10) {$q_{j+9}$};
\node[state, right of=16,fill=gray!30] (8) {$q_{i+8}$};
\node[state, above of=5] (17) {$q_{j+7}$};
\node[state, above of=4] (18) {$q_{j+6}$};
\node[state, above of=16] (19) {$q_{j+10}$};
\node[state, above of=8] (20) {$q_{j+11}$};
\node[below of=8, node distance=1.3cm] (11) {};
\node[state, above of=0] (12)  {$q_{j}$};
\node[state, above of=1] (13)  {$q_{j+1}$};
\node[state, above of=2] (14)  {$q_{j+2}$};
\node[state, above of=3] (15)  {$q_{j+3}$};

\path
(-1) edge[dotted] (0)
	(0) edge node[above] {$\triangleright$} (1)
	(1) edge node[above] {$\ta$} (2)
	(2) edge node {$\#$} (3)
	(3) edge node {$\triangleright$} (4)
	(4) edge node {$\ta$} (5)
	(5) edge node {$\ta$} (6)
	(6) edge node {$\ta$} (7)
	(7) edge node {$\ta$} (16)
	(16) edge node {$\$$} (8)
	(3) edge[below right] node {$\triangleright$} (18)
	(18) edge node {$\ta$} (17)
	(17) edge node {$\ta$} (9)
	(9) edge node {$\ta$} (10)
	(10) edge node {$\ta$} (19)
	(19) edge node {$\ta$} (20)
	(20) edge node[right] {$\$$} (8)
(8) edge[dotted] (11)
(0) edge node[right] {$\triangleright$} (12)
(12) edge node[above] {$\ta$} (13)
(13) edge node[above] {$\ta$} (14)
(14) edge node[above] {$\ta$} (15)
(15) edge node[right] {$\#$} (3)
;
\end{tikzpicture}
}
\caption{Examples of subautomata of $M$, $M_{\rep_M}$, and $\widehat{M}_{\rep_M}$ for an NFA $M$ with $\lang(M) \subseteq \Enc$. The top subautomaton shows a part of $M$ where left vertex tokens can be read between $q_i$ and $q_{i+3}$ and right vertex tokens can be read between $q_{i+3}$ and $q_{i+8}$. The subautomaton in the middle shows the corresponding part in $M_{\rep_M}$ where $\rep_M(q_i, q_{i+3}) = \{\triangleright \ta \#$, $ \triangleright \ta\ta\ta\ta\#\}, \rep_M(q_{i+3}, q_{i+8})= \{\triangleright \ta \ta \ta \ta \$$, $ \triangleright \ta \ta \ta \ta \ta \$\}$ and $\rep_M(p, p') = \emptyset$ for all other pairs of drawn states $p, p'$.
The bottom subautomaton shows the generalized NFA $\widehat{M}_{\rep_M}$ obtained from ${M}_{\rep_M}$ where all representatives are interpreted over the original alphabet of $M$ (note that this leads to additional states). 
}
\label{fig:MrepM}
\end{figure}


\begin{figure}[H]
	\centering
	\begin{tabular}{c|c}
		Lemma~\ref{lem:connectionMerge} & Lemma~\ref{lem:connectionSeparate}\\
		\scalebox{0.7}{
	\begin{tabular}{ccc}
		\begin{tikzpicture}
		\node[fill, circle, inner sep = 3pt] (1) at (-4.5,1) {};
		\node[fill, circle, inner sep = 3pt] (2) at (-4.5,-1) {};
		\node[state, minimum size=20pt, fill=white] (3) at (-3.8,0) {$u$};
		\node[state, minimum size=20pt, fill=white] (4) at (-2.3,0) {$v$};
		\node[fill, circle, inner sep = 3pt] (5) at (-1.5,1) {};
		\node[fill, circle, inner sep = 3pt] (6) at (-1.5,-1) {};
		\path[draw,line width=1pt]
		(1) edge (3)
		(2) edge (3)
		(4) edge (5)
		(4) edge (6);		
		\end{tikzpicture}&
		\begin{tikzpicture}
		\node[shape=dart, draw=black, fill=gray, minimum size=6mm] (z) at (0,1) {};
		\node (0,-1) {};
		\end{tikzpicture}&
		\begin{tikzpicture}
		\node[fill, circle, inner sep = 3pt] (7) at (2.5,1) {};
		\node[fill, circle, inner sep = 3pt] (8) at (2.5,-1) {};
		\node[state, minimum size=20pt, fill=white] (9) at (3.5,0) {$[u,v]$};
		\node[fill, circle, inner sep = 3pt] (10) at (4.5,1) {};
		\node[fill, circle, inner sep = 3pt] (11) at (4.5,-1) {};
		
		\path[draw,line width=1pt]
		(7) edge (9)
		(8) edge (9)
		(10) edge (9)
		(11) edge (9);
		\end{tikzpicture}\\
		& & \\
		\begin{tikzpicture}
		\node[fill, circle, inner sep = 3pt] (3) at (-4,0) {};
		\node[state, minimum size=20pt, fill=white] (4) at (-2.5,-1) {$v$};
		\node[state, minimum size=20pt, fill=white] (5) at (-2,0) {$u$};
		\node[fill, circle, inner sep = 3pt] (6) at (-2.5,1) {};
			
		\path[draw,line width=1pt]
		(3) edge (4)
		(3) edge (5)
		(3) edge (6);	
		\end{tikzpicture}&
		\begin{tikzpicture}
		\node[shape=dart, draw=black, fill=gray, minimum size=6mm] (z) at (0,1) {};
		\node (0,-1) {};
		\end{tikzpicture}&
		\begin{tikzpicture}
		\node[fill, circle, inner sep = 3pt] (7) at (2,0) {};
		\node[state, minimum size=20pt, fill=white] (8) at (4,0) {$[u,v]$};
		\node[fill, circle, inner sep = 3pt] (9) at (3.5,1) {};
		\node(4) at (2,-1) {};
		
		\path[draw,line width=1pt]
		(8) edge (7)
		(7) edge (9);
		\begin{scope}[on background layer]
		\path[draw,line width=1pt]
		(2,-0.15) edge[dotted] (4.5, -0.15);
		\end{scope}
	\end{tikzpicture}\\
	& & \\
	\begin{tikzpicture}
	\node[state, minimum size=20pt, fill=white] (4) at (2,0) {$v$};
	\node[state, minimum size=20pt, fill=white] (5) at (0,0) {$u$};
	
	\path[draw,line width=1pt]
	(4) edge (5);	
	\end{tikzpicture}&
	\begin{tikzpicture}
	\node[shape=dart, draw=black, fill=gray, minimum size=6mm] (z) at (0,1) {};
	\end{tikzpicture}&
	\begin{tikzpicture}
	\node[state, minimum size=20pt, fill=white] (8) at (0,0) {$[u,v]$};
	
	\path[draw,line width=1pt]
	(8) edge[loop right,dotted] (8);
	\end{tikzpicture}\\
\end{tabular}
}
&
\scalebox{0.7}{
\begin{tabular}{ccc}\\
	& & \\
	\begin{tikzpicture}
	\node[fill, circle, inner sep = 3pt] (1) at (-4.5,1) {};
	\node[fill, circle, inner sep = 3pt] (2) at (-4.5,-1) {};
	\node[state, minimum size=20pt, fill=white] (3) at (-3.8,0) {$u$};
	\node[state, minimum size=20pt, fill=white] (4) at (-2.3,0) {$v$};
	\node[fill, circle, inner sep = 3pt] (5) at (-1.5,1) {};
	\node[fill, circle, inner sep = 3pt] (6) at (-1.5,-1) {};
	\path[draw,line width=1pt]
	(1) edge (3)
	(2) edge (3)
	(4) edge (5)
	(4) edge (6)
	(3) edge (4);		
	\end{tikzpicture}&
	\begin{tikzpicture}
	\node[shape=dart, draw=black, fill=gray, minimum size=6mm] (z) at (0,1) {};
	\node (0,-1) {};
	\end{tikzpicture}&
	\begin{tikzpicture}
	\node[fill, circle, inner sep = 3pt] (1) at (-4.5,1) {};
	\node[fill, circle, inner sep = 3pt] (2) at (-4.5,-1) {};
	\node[state, minimum size=20pt, fill=white] (3) at (-3.8,0) {$u$};
	\node[state, minimum size=20pt, fill=white] (4) at (-1.3,0) {$v$};
	\node[fill, circle, inner sep = 3pt] (5) at (-0.5,1) {};
	\node[fill, circle, inner sep = 3pt] (6) at (-0.5,-1) {};
	\node[state, minimum size=20pt, fill=white] (7) at (-2.5,0.9) {$v'$};
	\path[draw,line width=1pt]
	(1) edge (3)
	(2) edge (3)
	(4) edge (5)
	(4) edge (6)
	(3) edge (7);		
	\end{tikzpicture}\\
	& & \\
	\begin{tikzpicture}
	\node[state, minimum size=20pt, fill=white] (8) at (0,0) {$u$};
	\node[state, minimum size=20pt, fill=white] (1) at (2,0) {$v$};
	\path[draw,line width=1pt]
	(1) edge[loop above,dotted] (1)
	(8) edge (1);
	\end{tikzpicture}&
	\begin{tikzpicture}
	\node[shape=dart, draw=black, fill=gray, minimum size=6mm] (z) at (0,1) {};
	\end{tikzpicture}&
	\begin{tikzpicture}
	\node[state, minimum size=20pt, fill=white] (7) at (2,0) {$u$};
	\node[state, minimum size=20pt, fill=white] (8) at (4.5,0) {$v$};
	\node[state, minimum size=20pt, fill=white] (9) at (3.5,1) {$v'$};
	
	\path[draw,line width=1pt]
	(8) edge[loop above,dotted] (8)
	(7) edge (9);
	\end{tikzpicture}\\
	& & \\
	\begin{tikzpicture}
	\node[state, minimum size=20pt, fill=white] (8) at (0,0) {$u$};
	\node[state, minimum size=20pt, fill=white] (1) at (2,0) {$v$};
	\path[draw,line width=1pt]
	(8) edge (1);
	\end{tikzpicture}&
	\begin{tikzpicture}
	\node[shape=dart, draw=black, fill=gray, minimum size=6mm] (z) at (0,1) {};
	\end{tikzpicture}&
	\begin{tikzpicture}
	\node[state, minimum size=20pt, fill=white] (7) at (2,0) {$u$};
	\node[state, minimum size=20pt, fill=white] (8) at (3.5,1) {$v'$};
	
	\path[draw,line width=1pt]
	(7) edge (8);
	\end{tikzpicture}\\
	& & \\
	\begin{tikzpicture}
	\node[state, minimum size=20pt, fill=white] (4) at (2,0) {$v$};
	\node[state, minimum size=20pt, fill=white] (5) at (0,0) {$u$};
	\begin{scope}[on background layer]
	\path[draw,line width=1pt]
		(0,0.15) edge (2,0.15)
		(0,-0.15) edge[dotted] (2,-0.15);
	\end{scope}
	\end{tikzpicture}&
	\begin{tikzpicture}
	\node[shape=dart, draw=black, fill=gray, minimum size=6mm] (z) at (0,1) {};
	\end{tikzpicture}&
	\begin{tikzpicture}
	\node[state, minimum size=20pt, fill=white] (7) at (2,0) {$u$};
	\node[state, minimum size=20pt, fill=white] (8) at (4.5,0) {$v$};
	\node[state, minimum size=20pt, fill=white] (9) at (4,1) {$v'$};
	
	\path[draw,line width=1pt]
	(8) edge (7)
	(7) edge (9);
	\end{tikzpicture}\\
	& & \\
	\begin{tikzpicture}
	\node[state, minimum size=20pt, fill=white] (8) at (0,0) {$v$};
	
	\path[draw,line width=1pt]
	(8) edge[loop right,dotted] (8);
	\end{tikzpicture}&
	\begin{tikzpicture}
	\node[shape=dart, draw=black, fill=gray, minimum size=6mm] (z) at (0,1) {};
	\end{tikzpicture}&
	\begin{tikzpicture}
	\node[state, minimum size=20pt, fill=white] (7) at (2,0) {$v$};
	\node[state, minimum size=20pt, fill=white] (8) at (4.5,0) {$v'$};
	\node (9) at (4,1) {};
	\path[draw,line width=1pt]
	(8) edge (7);
	\end{tikzpicture}\\
	& & \\
	\begin{tikzpicture}
	\node[state, minimum size=20pt, fill=white] (8) at (0,0) {$v$};
	\node[fill, circle, inner sep = 3pt] (1) at (2,0) {};
	\path[draw,line width=1pt]
	(8) edge[loop above,dotted] (8)
	(8) edge (1);
	\end{tikzpicture}&
	\begin{tikzpicture}
	\node[shape=dart, draw=black, fill=gray, minimum size=6mm] (z) at (0,1) {};
	\end{tikzpicture}&
	\begin{tikzpicture}
	\node[state, minimum size=20pt, fill=white] (7) at (2,0) {$v$};
	\node[fill, circle, inner sep = 3pt] (8) at (4.5,0) {};
	\node[state, minimum size=20pt, fill=white] (9) at (4,1) {$v'$};
	
	\path[draw,line width=1pt]
	(8) edge (7)
	(7) edge (9);
	\end{tikzpicture}\\
	& & \\
\end{tabular}
}
\end{tabular}
\caption{Connection between replacing a vertex token by a representative and the performed graph operation on the encoded graph. On the left-hand side, all vertex tokens encoding the vertex $v$ are replaced by representatives from $\pickmerge_V$ such that they fall together with the representatives encoding $u$. On the right-hand side, a vertex token $u_i$ encoding the vertex $v$ is replaced by a representative from $\picksep_V$. Edge factors in the word which are omitted by our decoding function are indicated by dotted lines.}
\label{fig:con-graph-rep}
\end{figure}
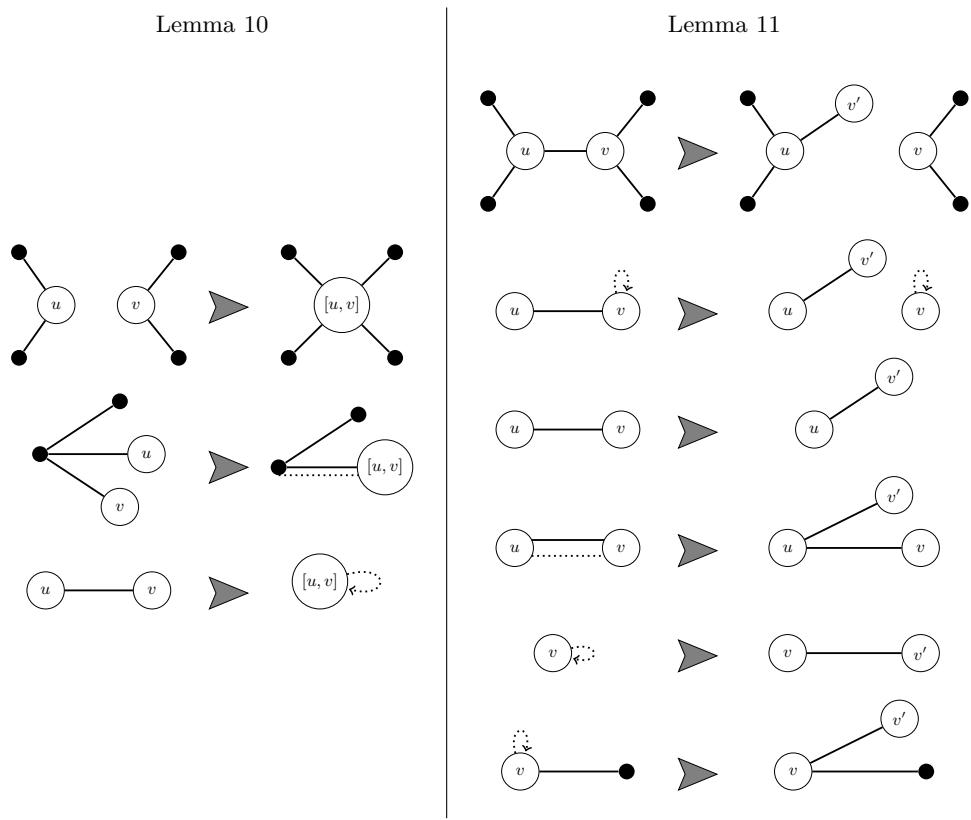
\begin{figure}[H]
	\centering
	\begin{tikzpicture}
	\node[state, minimum size=20pt, fill=white] (1) at (0,0) {$u$};
	\node[state, minimum size=20pt, fill=white] (2) at (1.5,0) {$v$};
	
	\node[shape=dart, draw=black, fill=gray, minimum size=4mm] (z) at (2.5,0.5) {};
	
	\node[state, minimum size=20pt, fill=white] (3) at (3.5,0) {$u$};
	\node[state, minimum size=20pt, fill=white] (4) at (5,0) {$v$};
	\node[state, minimum size=20pt, fill=white] (5) at (5,1) {$v'$};
	
	\node[shape=dart, draw=black, fill=gray, minimum size=4mm] (z) at (6,0.5) {};
	
	\node[state, minimum size=20pt, fill=white] (6) at (7,0) {$u$};
	\node[state, minimum size=20pt, fill=white] (7) at (8.5,0) {$v$};
	\node[state, minimum size=20pt, fill=white] (8) at (7,1) {$u'$};
	\node[state, minimum size=20pt, fill=white] (9) at (8.5,1) {$v'$};
	
	\begin{scope}[on background layer]
	\path[draw,line width=1pt]
	(0,0.15) edge (1.5,0.15)
	(0,-0.15) edge[dotted] (1.5,-0.15)
	(3) edge (4)
	(3) edge (5)
	(6) edge (7)
	(8) edge (9);
	\end{scope}
	\end{tikzpicture}
	\caption{Applying an add-leaf operation followed by a separate operation corresponds to separating a multi-edge.}
	\label{fig:graph-ops-combi}
\end{figure}
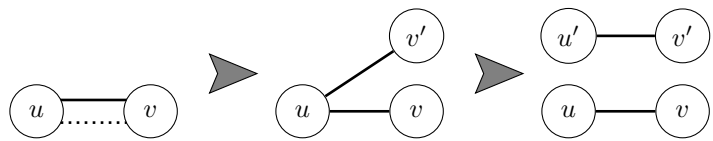

\begin{figure}[H]
	\centering
		\begin{tikzpicture}
	{
			{\node[fill, circle, inner sep = 3pt] (1) at (0,0) {};}
			{\node[fill, circle, inner sep = 3pt] (2) at (1,0.5) {};}
			{\node[fill, circle, inner sep = 3pt] (3) at (2.5,0.3) {};}
			{\node[fill, circle, inner sep = 3pt] (4) at (3.5, -0.2) {};}

			\path[draw,line width=1pt]
			(1) edge (2)
			(3) edge (4);
		}
		
		{
			\draw [draw=gray, line width=3][->] (4.5,0.5) --(6,0.5);
			\node (9) at (5.2,1) {Merging};
			{\node[fill, circle, inner sep = 3pt] (5) at (7,0) {};}
			{\node[fill, circle, inner sep = 3pt] (6) at (8,0.5) {};}
			{\node[fill, circle, inner sep = 3pt] (7) at (9, -0.2) {};}
			\path[draw,line width=1pt]
			(5) edge (6)
			(6) edge (7);
		}
	{
		\draw (1.5,-1) node [rounded corners=8pt, fill=white] {$a\#aaaaaa\$\underline{aa}\#aaa\$$};
	}
	{
		\draw (8,-1) node [rounded corners=8pt, fill=white] {$a\#{aaaaaa}\${aaaaaa}\#aaa\$$};
	}
	\end{tikzpicture}
	
%
%
\vspace{1cm}
	\centering

		\begin{tikzpicture}
{
			{\node[fill, circle, inner sep = 3pt] (1) at (6,0) {};}
			{\node[fill, circle, inner sep = 3pt] (2) at (7,0.5) {};}
			{\node[fill, circle, inner sep = 3pt] (3) at (8.5,0.3) {};}
			{\node[fill, circle, inner sep = 3pt] (4) at (9.5, -0.2) {};}
			
			\path[draw,line width=1pt]
			(1) edge (2)
			(3) edge (4);
		}
{
			\draw [draw=gray, line width=3][->] (3.5,0.5) --(5,0.5);
			
			\node (9) at (4.2,1) {Separating};
			{\node[fill, circle, inner sep = 3pt] (5) at (0,0) {};}
			{\node[fill, circle, inner sep = 3pt] (6) at (1,0.5) {};}
			{\node[fill, circle, inner sep = 3pt] (7) at (2, -0.2) {};}
			\path[draw,line width=1pt]
			(5) edge (6)
			(6) edge (7);
		}

{
		\draw (1,-1) node [rounded corners=8pt, fill=white] {$a\#aa\$\underline{aa}\#aaa\$$};
	}
{
		\draw (7.5,-1) node [rounded corners=8pt, fill=white] {$a\#{aa}\${aaaaaaaa}\#aaa\$$};
	}
	\end{tikzpicture}
	\caption{Pumping labels of vertices and its impact on the encoded graph.}
	\label{fig:beamer-fig}
\end{figure}
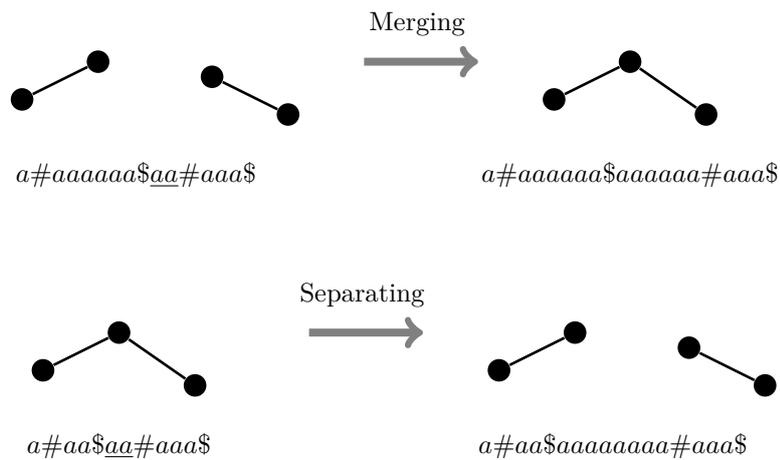
\section{Proofs for Section 4 (Applications -- Decidability Results)}

\begin{proof}[Proof of Lemma~\ref{lem:mergeVC}]
	Since $G$ has a vertex cover, for each edge there has to be one vertex in $V'$. We make a case distinction of the vertices contained in $V'$.
	\begin{itemize}
		\item If $v \notin V'$ and $v' \notin V'$, then for all $\{v, u\}, \{v', u'\} \in E$ the vertices $u$ and $u'$ has to be in $V'$. If $v$ and $v'$ are merged, all edges $\{[v,v'], u \} \in E^\star$ are also covered by $V'' = V'$.
		\item If at least one of $v$ and $v'$ is contained in $V'$, then $V'' = (V'\backslash \{v, v'\}) \cup \{[v,v']\}$ is a vertex cover of size at most $k$.
		\item If $v \in V'$, $v' \in V'$, and $v$ and $v'$ are merged, edges $\{v, u \} \in E$ are replaced by edges $\{[v,v'], u\} \in E^\star$ and edges $\{v', u' \} \in E$ are replaced by edges $\{[v,v'], u'\} \in E^\star$. Hence, both types of edges in $E^\star$ are covered by $V'' = (V' \backslash \{v',v\}) \cup \{[v,v']\}$. 
	\end{itemize}
	Edges not containing $v$ or $v'$ have not been changed and hence $V''$ is a vertex cover for $G^\star$ of size at most $k$. 
\end{proof}

\begin{proof}[Proof of Lemma~\ref{lem:ISreplace}]
	By assumption, $G$ contains an IS $V'$ of size at least $k$. By construction of  $G^\star$ and $G^\diamond$, $V'$ is also independent in those graphs.
\end{proof}
\end{document}